\documentclass[twocolumn,10pt]{article}
\usepackage[T1]{fontenc}
\usepackage[utf8]{inputenc}
\usepackage{mathptmx}
\usepackage[scaled=.90]{helvet}
\usepackage{courier}

\usepackage[a4paper,top=0.6in,bottom=0.75in,left=0.6in,right=0.6in,includefoot]{geometry}

\usepackage{graphicx}
\usepackage{xcolor}
\usepackage{cuted}
\usepackage{microtype}
\usepackage{parskip}
\usepackage{tcolorbox}
\tcbuselibrary{skins,breakable}
\usepackage{academicons}
\usepackage{fontawesome5}
\usepackage{hyperref}
\usepackage{lipsum}
\usepackage{lastpage}
\usepackage{tikz}
\usetikzlibrary{shapes.geometric,arrows.meta,positioning,fit,backgrounds,calc}
\usepackage{booktabs}
\usepackage{tabularx}
\usepackage{caption}
\usepackage{enumitem}
\usepackage{fancyhdr}
\usepackage{titlesec}
\usepackage{float}
\usepackage[ruled,noline]{algorithm2e}
\usepackage{wrapfig}
\usepackage{amsmath, amssymb}
\usepackage{mathtools}
\usepackage{multicol}
\usepackage{multirow}
\usepackage{amsthm}

\newtheorem{proposition}{Proposition}

\newtheorem{definition}{Definition}

{\theoremstyle{remark}\newtheorem{remark}{Remark}}
{\theoremstyle{remark}}

\usepackage{cite}
\usepackage{stfloats}

\graphicspath{{figs/}}

\definecolor{primary}{HTML}{003049}
\definecolor{accent}{HTML}{F77F00}
\definecolor{textgray}{HTML}{2F2F2F}
\definecolor{dividergray}{HTML}{DADADA}

\hypersetup{
  colorlinks=true,
  linkcolor=primary,
  citecolor=primary,
  urlcolor=primary
}

\titleformat{\section}
  {\color{primary}\sffamily\Large\bfseries}
  {\thesection}{1em}{}[\vspace{0.2em}\titlerule]

\titleformat{\subsection}
  {\color{primary}\sffamily\large\bfseries}
  {\thesubsection}{1em}{}

\titlespacing*{\section}{0pt}{6pt}{3pt}
\titlespacing*{\subsection}{0pt}{4pt}{2pt}

\setlist[itemize]{left=1.2em, itemsep=2pt, topsep=2pt, parsep=0pt}
\setlist[enumerate]{left=1.2em, itemsep=2pt, topsep=2pt, parsep=0pt}

\tcbset{
  abstractstyle/.style={
    enhanced,
    colback=white,
    colframe=primary,
    boxrule=1pt,
    fonttitle=\bfseries\sffamily\large,
    left=6mm, right=6mm, top=2mm, bottom=2mm,
    width=\textwidth,
    boxsep=4pt,
    breakable
  },
  biobox/.style={
    colback=white,
    colframe=primary,
    arc=1.5mm,
    boxrule=0.5pt,
    drop shadow southeast,
    left=3mm, right=3mm, top=1mm, bottom=1mm,
    width=\linewidth,
    before skip=6pt, after skip=6pt,
    breakable
  }
}

\SetAlgoNlRelativeSize{-1}
\SetAlCapNameFnt{\sffamily\bfseries\small}
\SetAlCapFnt{\sffamily\small}

\makeatletter
\def\@maketitle{%
  \begin{center}
    {\fontsize{22pt}{18pt}\selectfont \bfseries \textcolor{primary}{Multi-Level Distributional Entropy from Flow Summary Statistics\\[4pt] for Explainable Network Intrusion Detection}}\\[1ex]
    {\normalsize
      Mohamed Aly Bouke\,%
      \raisebox{0.6ex}{\href{https://orcid.org/0000-0003-3264-601X}{\includegraphics[height=1.4ex]{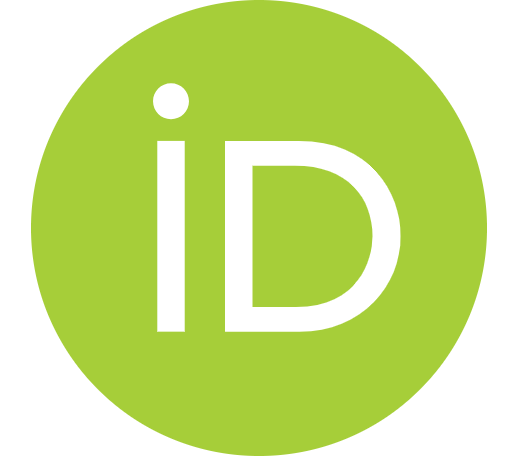}}}%
      \,\raisebox{0.6ex}{\href{mailto:bouke@ieee.org}{\textcolor{gray}{\scriptsize\faEnvelope}}}%
      \raisebox{1.3ex}{\scriptsize1,*},\enspace%
      Md Shohel Sayeed\,%
      \raisebox{0.6ex}{\href{https://orcid.org/0000-0002-0052-4870}{\includegraphics[height=1.4ex]{orcid.png}}}%
      \,\raisebox{0.6ex}{\href{mailto:shohel.sayeed@mmu.edu.my}{\textcolor{gray}{\scriptsize\faEnvelope}}}%
      \raisebox{1.3ex}{\scriptsize1},\enspace
      Swee-Huay Heng\,%
      \raisebox{0.6ex}{\href{https://orcid.org/0000-0003-3627-2131}{\includegraphics[height=1.4ex]{orcid.png}}}%
      \,\raisebox{0.6ex}{\href{mailto:shheng@mmu.edu.my}{\textcolor{gray}{\scriptsize\faEnvelope}}}%
      \raisebox{1.3ex}{\scriptsize1},\enspace
      Azizol Abdullah\,%
      \raisebox{0.6ex}{\href{https://orcid.org/0000-0001-8321-9259}{\includegraphics[height=1.4ex]{orcid.png}}}%
      \,\raisebox{0.6ex}{\href{mailto:azizol@upm.edu.my}{\textcolor{gray}{\scriptsize\faEnvelope}}}%
      \raisebox{1.3ex}{\scriptsize2},\enspace
      Mohamed Othman\,%
      \raisebox{0.6ex}{\href{https://orcid.org/0000-0002-5124-5759}{\includegraphics[height=1.4ex]{orcid.png}}}%
      \,\raisebox{0.6ex}{\href{mailto:mothman@upm.edu.my}{\textcolor{gray}{\scriptsize\faEnvelope}}}%
      \raisebox{1.3ex}{\scriptsize2,3}
    }\\[0.8ex]
    {\footnotesize
      \textsuperscript{1}Centre for Intelligent Cloud Computing, CoE for Advanced Cloud,\par
      Faculty of Information Science and Technology,\par
      Multimedia University, Jalan Ayer Keroh Lama, Bukit Beruang, 75450, Melaka, Malaysia\\[0.3ex]
      \textsuperscript{2}Department of Communication Technology and Networking,\par
      Faculty of Computer Science and Information Technology,\par
      Universiti Putra Malaysia, Serdang 43400, Malaysia\\[0.3ex]
      \textsuperscript{3}Laboratory of Computational Science and Mathematical Physics,\par
      Institute for Mathematical Research,\par
      Universiti Putra Malaysia, Serdang, Malaysia
    }\\[0.5ex]
    {\scriptsize \texttt{*bouke@ieee.org}}\\[1ex]
    {\scriptsize \textit{Research Article}, \today}
  \end{center}
}
\renewcommand{\maketitle}{%
  \twocolumn[
    \color{textgray}
    \@maketitle
    \vspace{-1.2em}
  ]
}
\makeatother

\begin{document}
\maketitle

\begin{strip}
  \begin{center}
    \begin{tcolorbox}[abstractstyle, title=Abstract]
      \normalsize
      Machine learning network intrusion detection systems (IDS) operate on aggregate flow statistics that discard the distributional structure of traffic, and although information-theoretic measures capture that structure, established entropy estimators require raw packet sequences that pre-aggregated flow datasets do not contain. No prior method derives entropy from the summary statistics those records already hold. We introduce Multi-Level Distributional Entropy (MDE), which computes interpretable information-theoretic features analytically from flow-level summary statistics at three levels, within-flow Gaussian differential entropy, cross-directional Jensen-Shannon divergence (JSD), and Transmission Control Protocol (TCP) flag-incidence Shannon entropy, with closed-form properties and no raw packet access; only imputation medians and score bounds are fitted on the training split. We pair the features with a leakage-free, fold-local evaluation protocol that reports the full operational metric suite across cross-validation, temporal, pseudo-live, cross-dataset, and unseen-attack-family settings, on four benchmarks (NSL-KDD, CICIDS-2017, CICIDS-2018, UNSW-NB15) with tree-ensemble classifiers and SHapley Additive exPlanations (SHAP). The protocol exposes failure modes that aggregate weighted F1 conceals: on CICIDS-2018 an F1 of 0.73 hides a detection rate (DR) of 0.44, on held-out attack families F1 exceeds 0.998 while DR falls to zero, and a 703K-flow pseudo-live replay reveals a threshold-ranking divergence in which score ranking is largely preserved (area under the ROC curve, AUC, 0.84 to 0.86) while fixed-threshold detection collapses (DR 0.08). The entropy features match conventional features within 0.1 percentage points of F1 and receive reproducible SHAP attributions (Spearman 0.84 to 0.94), so their contribution is a grounded, interpretable representation and an evaluation methodology rather than an accuracy gain.

      \vspace{0.5em}
    \end{tcolorbox}
  \end{center}
\end{strip}

\section{Introduction}\label{sec:intro}
\vspace{0.5em}

The proliferation of networked infrastructure across enterprise, cloud, and embedded environments has made automated intrusion detection an operational necessity. Machine learning-based IDS have matured substantially over the past decade, with ensemble classifiers, particularly gradient-boosted trees and random forests, consistently achieving strong detection accuracy on standard benchmarks~\cite{khraisat2019survey,ring2019survey}. Despite this progress, three gaps in current IDS practice constrain broader applicability.

First, standard IDS pipelines operate on flow-level feature vectors: packet counts, byte volumes, connection duration, and statistical summaries extracted by tools such as CICFlowMeter. These aggregate descriptors discard the distributional and sequential structure of traffic that is known to differ systematically between benign sessions and attacks~\cite{bouke2026entropy}. Information-theoretic measures quantify exactly this structural difference, yet their integration into supervised IDS pipelines has been limited by a practical barrier. Classical entropy computation such as approximate entropy (ApEn)~\cite{pincus1991approximate} and sample entropy (SampEn)~\cite{richman2000sample} requires ordered sequences of raw packet measurements that are unavailable in pre-aggregated flow datasets, the dominant format in the IDS benchmark literature~\cite{tavallaee2009kdd,sharafaldin2018cicids,moustafa2015unsw}. No prior work derives entropy analytically from the summary statistics that flow records already contain.

Second, interpretability remains limited. High-performing ensemble classifiers are opaque, and in operational security environments analysts must audit model decisions before acting on them; in regulated contexts this opacity creates compliance obstacles~\cite{bouke2025xairf,bouke2026lightgbm}. SHAP provides a theoretically grounded attribution framework~\cite{lundberg2017unified,lundberg2018treeshap}, but its application to entropy-enriched IDS pipelines has not been studied, leaving open whether entropy features receive domain-coherent SHAP explanations or produce consistent attributions across environments.

Third, entropy-based detectors are typically evaluated on a single dataset, so the cross-dataset transferability of entropy signatures and the factors governing it remain unstudied.

This paper addresses these gaps through the MDE framework. The statistics already present in flow records (means, standard deviations, minima, maxima, and packet counts) implicitly characterize the underlying distributions of packet sizes and inter-arrival times; MDE computes entropy analytically from those statistics, eliminating the need for raw packet access while producing features grounded in information theory and directly interpretable via SHAP.

MDE differs from three categories of related work. \emph{Conventional flow statistics} (byte counts, packet rates, duration) capture aggregate magnitudes but not distributional structure, and are susceptible to the labeling artifacts documented by Engelen et al.~\cite{engelen2021troubleshooting}. \emph{Statistical feature engineering} approaches (e.g.\ skewness, kurtosis, higher-order moments) are empirically motivated but lack a principled basis for feature selection; MDE features are grounded in information theory with analytic definitions and known ranges (Propositions~\ref{prop:ade} and~\ref{prop:jsd}). \emph{Representation learning} methods (convolutional neural networks (CNNs), long short-term memory networks (LSTMs), autoencoders) extract features through training and produce non-interpretable representations; MDE features are deterministic functions of the flow statistics with no learned parameters, can be instantiated on flow-level schemas that record directional statistics, and are natively interpretable via SHAP.

The main contributions of this paper are as follows:
\begin{enumerate}
  \item We propose MDE, a method that constructs 7--12 interpretable features \emph{analytically} from pre-aggregated network flow statistics at three complementary entropy levels, within-flow Gaussian differential entropy (L1), cross-directional JSD (L2), and flag-incidence Shannon entropy (L3), with closed-form properties and without raw packet access.
  \item We develop a leakage-free, fold-local protocol that reports the full operational metric suite (DR, false alarm rate (FAR), Matthews Correlation Coefficient (MCC), and precision-recall AUC (PR-AUC) alongside F1) and applies it across cross-validation, temporal-split, pseudo-live replay, cross-dataset transfer, and unseen-attack-family settings, designed to surface failure modes that aggregate scores conceal.
  \item We analyze SHAP attributions for MDE-augmented classifiers across five cross-validation folds, quantifying the rank stability of the attributions and the share of the top-ranked features that the entropy levels receive across structurally distinct environments.
\end{enumerate}

The remainder is organized as follows. Section~\ref{sec:related} reviews related work. Section~\ref{sec:background} covers entropy and SHAP foundations. Section~\ref{sec:framework} presents the MDE framework. Section~\ref{sec:setup} describes the experimental setup. Section~\ref{sec:results} reports the results. Sections~\ref{sec:discussion} and~\ref{sec:conclusion} discuss and conclude.

\section{Related Work}\label{sec:related}
\vspace{0.5em}

Work bearing on MDE falls into three strands: the detection architectures that dominate flow-based intrusion detection, the information-theoretic features that motivate an entropy representation, and the explainability methods now expected in operational settings. Each is reviewed in turn below, and a final subsection positions MDE against the gap they jointly leave open.

\subsection{Detection Architectures}

ML-based IDS has evolved from early shallow classifiers through gradient-boosted ensembles to transformer, graph neural network (GNN), and federated architectures~\cite{khraisat2019survey,thakkar2022survey,abdulganiyu2023systematic}. Gradient-boosted trees and random forests remain consistently dominant on tabular flow data~\cite{thakkar2022survey}, while deep learning adds value in sequential and raw-packet settings. Transformer architectures have been applied to flow-level sequences~\cite{wu2022rtids}, and a recent comprehensive survey covering attention-based and large language model approaches~\cite{kheddar2025transformers} identifies cross-dataset generalization as the principal unsolved challenge. GNN-based approaches model network topology explicitly and show promise for lateral-movement and Advanced Persistent Threat (APT) detection~\cite{zhong2024gnn}, though they require graph construction infrastructure absent in standard flow deployments. Federated IDS~\cite{buyuktanir2025federated} address cross-organizational data-sharing constraints, achieving near-centralized accuracy while preserving local privacy, which is an increasingly important practical requirement. Despite these architectural advances, the literature consistently relies on the same families of aggregate flow statistics (packet counts, byte volumes, flow duration) without principled information-theoretic enrichment~\cite{ring2019survey,apruzzese2022sok}. Tama and Rhee~\cite{tama2019indepth} and Faker and Dogdu~\cite{faker2019intrusion} establish strong baselines on NSL-KDD with XGBoost and deep learning respectively, but both operate on entirely conventional feature sets. The recurring limitation is therefore one of representation rather than architecture. Effort has concentrated on the classifier, while the input feature space has remained aggregate and blind to the distributional structure of traffic.

Beyond representation, the reliability of IDS evaluation is itself contested. Engelen et al.~\cite{engelen2021troubleshooting} audited CICIDS-2017 and found labeling artifacts making some attack categories trivially separable by a single feature (e.g.\ \texttt{Flow Duration}$=0$). Ring et al.~\cite{ring2019survey} surveyed 34 IDS datasets and concluded that most published accuracy figures are inflated by dataset-specific artifacts. Sarhan et al.~\cite{sarhan2021netflow} showed that cross-dataset transfer typically collapses to near-chance performance. Recent work on standardized dataset evaluation frameworks~\cite{tori2025framework} proposes MITRE ATT\&CK-aligned metrics for assessing dataset relevance to real threat scenarios, reinforcing that single-dataset evaluations are insufficient for claiming generalization. Taken together, these findings indicate that single-dataset, within-distribution cross-validation (CV) evaluations reported through aggregate scores systematically overstate operational performance, yet such protocols remain the field norm.

\subsection{Entropy-Based Features}

Information-theoretic measures have a foundational history in network anomaly detection. Wagner and Plattner~\cite{wagner2005entropy} showed that Shannon entropy of packet-size and source-address distributions detects worm propagation and Denial-of-Service (DoS) floods, establishing that entropy captures the structural regularity of benign traffic more reliably than absolute counts. Nychis et al.~\cite{nychis2008empirical} confirmed empirically that header-field entropy is an effective anomaly indicator across multiple operational networks. Xu et al.~\cite{xu2005internet} showed that flooding attacks collapse destination-port entropy while benign traffic maintains it across all dimensions. Recent work continues in this direction. Kenyon~\cite{kenyon2024entropy} characterizes payload entropy profiles across flow types, confirming that entropy measurements reliably distinguish attack and benign flows even without payload decryption; Yu et al.~\cite{yu2024renyi} propose Rényi-entropy-driven anomaly detection with dynamic thresholding, achieving lower false-alarm rates than fixed-threshold Shannon-entropy detectors. Sequential complexity measures such as ApEn~\cite{pincus1991approximate} and SampEn~\cite{richman2000sample} require ordered packet sequences that are incompatible with pre-aggregated CICFlowMeter output~\cite{sharafaldin2018cicids,sarhan2021netflow}. A gap therefore persists between the well-documented discriminative value of entropy and the pre-aggregated flow formats that dominate practical IDS pipelines, for which no existing estimator derives entropy from the summary statistics the records already contain.

\subsection{Explainability}

Explainability has become an operational requirement for IDS in high-stakes and regulated environments. Mahbooba et al.~\cite{mahbooba2021explainable} showed that rule-based explanations raise analyst alert-validation rates from 61\% to 94\%. Patil et al.~\cite{patil2022explainable} demonstrated that SHAP attribution improves analyst trust in cloud-based IDS. A recent systematic review~\cite{khan2025xai} covering explainable AI (XAI) methods applied to IDS (2020--2024) concludes that TreeSHAP is the dominant explanation approach for ensemble classifiers in industrial deployments, valued for its exact computation and consistency guarantees~\cite{lundberg2018treeshap}. Bouke et al.~\cite{bouke2025xairf,bouke2026lightgbm,bouke2024bukagini} have applied SHAP to spam detection and IDS, confirming that tree-ensemble attributions are stable across folds and align with domain-expert intuition. These studies, however, attribute exclusively over conventional flow features; whether information-theoretic features receive coherent attributions, and whether those attributions remain consistent across heterogeneous environments, has not been examined.

\subsection{Synthesis and Positioning}

These three lines of work leave one need unaddressed. Architectural innovation in ML-based IDS has outpaced innovation in feature representation, which remains aggregate and structure-agnostic; entropy captures the distributional structure that aggregate statistics discard, but its established estimators are incompatible with the pre-aggregated datasets that dominate the field; explainability is now an operational requirement, yet entropy-derived features have never been subjected to attribution analysis; and prevailing single-dataset, aggregate-metric protocols overstate generalization.

The present work is positioned precisely at this intersection. It derives entropy analytically from the summary statistics already present in flow records, removing the raw-packet barrier while preserving the distributional sensitivity that motivates entropy. Its features have explicit definitions and are natively interpretable, enabling a SHAP fold-stability analysis of entropy attributions. It is then evaluated under a leakage-free, multi-dataset protocol with full operational metrics that exposes the failure modes aggregate single-dataset scores conceal. The following sections formalize this framework (Section~\ref{sec:framework}) and evaluate it (Section~\ref{sec:results}).

\section{Background}\label{sec:background}
\vspace{0.5em}

The MDE framework rests on three information-theoretic constructs and one attribution framework, defined here in the form used throughout the paper.

\begin{definition}[Shannon entropy~\cite{shannon1948mathematical}]
For a discrete random variable $X$ with probability mass function $p$, the Shannon entropy is
$H(X) = -\sum_x p(x)\log p(x),$
where $\log$ denotes the natural logarithm throughout this paper (units: nats); $H(X) \geq 0$ with equality only when $X$ is deterministic.
\end{definition}

\begin{definition}[Gaussian differential entropy~\cite{cover1991elements}]
For a continuous random variable $X \sim \mathcal{N}(\mu,\sigma^2)$, the differential entropy is
$h(X) = \tfrac{1}{2}\ln(2\pi e\,\sigma^2),$
where $e \approx 2.718$ is Euler's number, $\mu$ is the mean, and $\sigma^2$ is the variance. $h(X)$ is a monotone increasing function of $\sigma$, and is independent of $\mu$.
\end{definition}

\begin{definition}[Jensen-Shannon divergence~\cite{lin1991divergence}]
For distributions $P$ and $Q$ with mixture $M = \tfrac{1}{2}(P+Q)$, the Jensen-Shannon divergence is
$\mathrm{JSD}(P\|Q) = \tfrac{1}{2}\mathrm{KL}(P\|M) + \tfrac{1}{2}\mathrm{KL}(Q\|M) \in [0,\ln 2],$
where $\mathrm{KL}(P\|Q)$ is the Kullback-Leibler divergence. JSD is symmetric and bounded, making it a well-suited measure of distributional asymmetry between forward and backward traffic.
\end{definition}

SHAP~\cite{lundberg2017unified} assigns each feature $i$ a value $\phi_i$ satisfying
$\hat{f}(x) = \phi_0 + \sum_{i=1}^{p}\phi_i,$
where $x \in \mathbb{R}^p$ is an input instance with $p$ features, $\hat{f}(x)$ is the model's scalar output for that instance, $\phi_0 = \mathbb{E}[\hat{f}]$ is the expected model output over the training set (the baseline prediction), and $\phi_i$ is the Shapley value from cooperative game theory, measuring the average marginal contribution of feature $i$ across all possible feature subsets. TreeSHAP~\cite{lundberg2018treeshap} computes exact Shapley values for tree ensembles in polynomial time, enabling efficient attribution across large feature sets.

\section{The MDE Framework}\label{sec:framework}
\vspace{0.5em}

MDE computes entropy analytically from the statistical summaries already present in pre-aggregated flow records, without requiring raw packet sequences. Fig.~\ref{fig:architecture} shows the end-to-end pipeline. Raw flow statistics feed three independent entropy levels (L1--L3), whose outputs are concatenated into a feature matrix that is passed to a tree-based classifier, with SHAP providing post-hoc per-prediction explanations. The framework is also dataset-aware, in that the specific feature columns used at each level adapt to the schema of each benchmark while the underlying analytical definitions stay uniform across all of them.

\subsection{Behavioral Representation}\label{sec:rationale}

MDE is a compact analytical feature construction method, not a feature selection method. Feature selection identifies which features within an existing set are most discriminative and discards the rest; the underlying feature pool is unchanged. MDE does not rank, filter, or reweight any existing flow feature. Instead, it constructs a small set of new features (entropy and divergence values) that are not present in any standard flow record and that did not exist before the transformation is applied. The inputs are the summary statistics already present in the record (means, standard deviations, packet counts, flag counts); the outputs are information-theoretic quantities grounded in a behavioral hypothesis about the structural difference between attack and benign traffic.

MDE is also distinct from general-purpose feature transformation approaches such as principal component analysis or kernel mappings. Those methods transform feature geometry without grounding the transformation in a domain theory. Each MDE feature corresponds to a specific, formalized behavioral claim: L1 captures within-flow distributional complexity (Proposition~\ref{prop:ade}); L2 captures directional asymmetry between source-to-destination and destination-to-source traffic (Proposition~\ref{prop:jsd}); L3 captures protocol-flag diversity. These are deterministic information-theoretic expressions, not learned representations, and their expected behavior under attack and benign conditions is analytically derivable from the underlying distributional assumptions.

\subsubsection{Feature Count}

The feature count follows directly from the three-level hierarchy applied to whichever directional statistics are available in each dataset schema. For each available direction (forward, backward) and signal type (packet size, inter-arrival time), one approximate differential entropy (ADE) value is computed wherever the required mean and standard deviation are present. One JSD value is computed for each directional pair whose marginal distributions can be parameterized from available statistics. One flag-entropy value is computed where per-flow flag counts are recorded. This yields 7 features on schemas with minimal directional statistics (NSL-KDD) and up to 12 on richer CICFlowMeter-format datasets. The count is not tuned to optimize any classification metric. It reflects the three behavioral dimensions (distributional complexity, directional asymmetry, and protocol irregularity) with minimal redundancy. Adding further features within the same hierarchy would introduce correlated quantities computed from the same input statistics.

\subsubsection{Analytical Properties}

The compact analytical design provides four concrete properties. First, \emph{portability}: the entropy functionals apply directly to any flow record that contains per-direction means, standard deviations, and counts, as the CICFlowMeter schema does, and on schemas that lack some of these inputs the levels are instantiated with proxies (Table~\ref{tab:inventory}), so the design carries over while the quantities are not identical. Second, \emph{interpretability}: each feature has an explicit definition and a direct theoretical connection to traffic behavior, so SHAP attributions reflect domain-relevant quantities rather than opaque numerical artifacts. Third, \emph{analytical tractability}: Propositions~\ref{prop:ade} and~\ref{prop:jsd} provide closed-form bounds that predict when the features will and will not be discriminative, enabling a priori reasoning about expected behavior before any model is trained. Fourth, \emph{bounded inference overhead}: each of the 7--12 MDE features is a deterministic function of the flow record, and its cost is measured in Section~\ref{sec:profiling}.

\subsubsection{Scope of the Contribution}

MDE's contribution is one of representation and evaluation methodology, not predictive uplift. It provides compact, theoretically grounded entropy features with a common three-level design, native SHAP interpretability, and full operational metric reporting, as evaluated in Section~\ref{sec:results}. Consistent with this scope, combined and conventional conditions achieve comparable within-distribution F1 on all tested datasets (Tables~\ref{tab:ablation} and~\ref{tab:fullmetrics}); MDE does not claim an accuracy gain over conventional flow statistics.

\subsection{Theoretical Foundations}\label{sec:theory}

Let $\mathcal{D} = \{(f_i,\, y_i)\}_{i=1}^N$ be a labeled IDS dataset where each flow
$f_i \in \mathbb{R}^d$ is a vector of pre-aggregated statistics and
$y_i \in \{0,1\}$ is a binary label ($y_i{=}1$: attack, $y_i{=}0$: benign).
The MDE transformation is a deterministic map $\phi: \mathbb{R}^d \to \mathbb{R}^k$
with $k \ll d$ that derives $k$ entropy-valued features directly from the statistics
already present in $f_i$.
The classifier operates on the full augmented space $\mathbb{R}^{d+k}$
(combined condition) or on $\mathbb{R}^k$ alone (entropy-only condition).

The discriminative power of MDE rests on a fundamental generative difference
between automated attack tools and human-driven applications.
Automated tools (scanners, Distributed Denial-of-Service (DDoS) agents, bots) operate under programmatic
constraints. They generate packets of fixed or narrowly bounded size, fire probes
at regular intervals, and produce near-unidirectional flows.
Human-driven sessions (browsing, Voice over Internet Protocol (VoIP), streaming) exhibit natural variability in
packet sizes, irregular inter-arrival timing, and bidirectional data exchange.
This asymmetry motivates three entropy-based discriminators:
\emph{within-flow variance} (L1), \emph{directional distributional asymmetry} (L2),
and \emph{protocol-flag diversity} (L3).

\begin{proposition}[Regularity--Entropy Correspondence]\label{prop:ade}
Let $X_a \sim \mathcal{N}(\mu_a,\sigma_a^2)$ and
$X_b \sim \mathcal{N}(\mu_b,\sigma_b^2)$ model the packet-size distributions of
an attack flow $a$ and a benign flow $b$ respectively.
If the attack tool imposes tighter size constraints than natural application
traffic, then $\sigma_a < \sigma_b$, and
\begin{equation}
  h(X_a)\;=\;\tfrac{1}{2}\ln(2\pi e\,\sigma_a^2)\;<\;\tfrac{1}{2}\ln(2\pi e\,\sigma_b^2)\;=\;h(X_b). \label{eq:prop1}
\end{equation}
\end{proposition}
\begin{proof}
$h(X)$ is strictly increasing in $\sigma^2$ since
$\frac{d}{d\sigma^2}\bigl[\tfrac{1}{2}\ln(2\pi e\,\sigma^2)\bigr]=\tfrac{1}{2\sigma^2}>0$.
Therefore $\sigma_a<\sigma_b \Rightarrow h(X_a)<h(X_b)$.
\end{proof}
\begin{remark}[What the Gaussian value measures]\label{rem:upper}
When the true packet-size distribution is not Gaussian, the quantity in Eq.~\ref{eq:prop1} keeps
a definite meaning but not the one of an estimate. By the maximum-entropy principle the Gaussian
maximizes differential entropy among all distributions with the same variance, so
$\tfrac{1}{2}\ln(2\pi e\,\sigma^2)$ is an upper bound on the true entropy,
$h(X) = \tfrac{1}{2}\ln(2\pi e\,\sigma^2) - J(X)$ with negentropy $J(X) \ge 0$. ADE is
therefore the entropy of the least informative distribution consistent with the reported variance,
a monotone transform of the variance, and the ordering in Proposition~\ref{prop:ade} is an ordering
of that surrogate. The entropy gap under the true distributions equals the ADE gap minus the
difference of the two negentropies, which has no fixed sign, so the bound does not by itself
transfer the ordering to real traffic. Whether the surrogate ordering separates attack from benign
flows is an empirical question, and Section~\ref{sec:results} treats it as one.
\end{remark}

\begin{remark}[Scope and limitations of the Gaussian approximation]
The Gaussian ADE (Eq.~\ref{eq:ade}) is an \emph{approximation} and \emph{not} a claim that network packet-size distributions are Gaussian. Real traffic exhibits heavy tails (skewness 1.5--9.7, excess kurtosis 0.9--110 across benchmarks) and may be multimodal (e.g.\ mixed HTTP/HTTPS flows). In such cases ADE overestimates the true differential entropy, since the Gaussian value is the maximum entropy attainable at the observed variance, and the approximation may fail for flows whose classification relies on distributional shape rather than variance contrast. The Gaussian assumption is expected to degrade in three settings, namely (1) multimodal traffic mixtures within a single flow, (2) datasets where attack and benign traffic share similar variance but differ in higher-order moments, and (3) encrypted traffic where packet payloads are padded to fixed sizes. Non-parametric alternatives (kernel-density entropy, Rényi/Tsallis entropy, histogram-based estimators) could address these cases but require design choices (bandwidth, order $\alpha$) that introduce their own model assumptions; we regard this as important future work.
\end{remark}

\begin{proposition}[Unidirectionality and Maximum JSD]\label{prop:jsd}
Let $P = \mathcal{N}(\mu_f,\sigma_f^2)$ and
$Q_\varepsilon = \mathcal{N}(\varepsilon\mu_f,\,\varepsilon^2\sigma_f^2)$
for $\varepsilon\in(0,1]$, modeling backward traffic scaled to fraction
$\varepsilon$ of the forward flow. Then:
\vspace{-0.3em}
\begin{enumerate}[label=(\roman*),nosep,leftmargin=*]
  \item $\mathrm{JSD}(P\|Q_1) = 0$ \enspace(symmetric bidirectional flow),
  \item $\displaystyle\lim_{\varepsilon\to 0^+}\!\mathrm{JSD}(P\|Q_\varepsilon)=\ln 2\approx 0.693$ \enspace(unidirectional flow).
\end{enumerate}
\end{proposition}
\begin{proof}[Proof sketch]
(i) is immediate since $P=Q_1$.
For (ii), let $M_\varepsilon = \tfrac{1}{2}(P+Q_\varepsilon)$ with densities $p$, $q_\varepsilon$, and
write $\mathrm{KL}(P\|M_\varepsilon) = \ln 2 - \int p\,\ln\!\bigl(1 + q_\varepsilon/p\bigr)\,dx$, and
likewise for $\mathrm{KL}(Q_\varepsilon\|M_\varepsilon)$. Both Gaussians have full support on $\mathbb{R}$,
so the argument rests on concentration rather than on disjoint supports. Fix $\delta>0$ and let
$A_\delta=[-\delta,\delta]$. As $\varepsilon\to0^+$, $Q_\varepsilon$ converges weakly to the point mass at
zero, so $Q_\varepsilon(A_\delta)\to1$, whereas $P(A_\delta)\to0$ as $\delta\to0$ because $P$ has a density.
Jensen's inequality on $A_\delta$ and on its complement gives
$\int p\ln(1+q_\varepsilon/p)\,dx \le P(A_\delta)\ln\!\bigl(1+1/P(A_\delta)\bigr)
+ P(A_\delta^{c})\ln\!\bigl(1+Q_\varepsilon(A_\delta^{c})/P(A_\delta^{c})\bigr)$; the second term vanishes as
$\varepsilon\to0^+$ and the first vanishes as $\delta\to0$. The same split bounds
$\int q_\varepsilon\ln(1+p/q_\varepsilon)\,dx$ by
$\ln\!\bigl(1+P(A_\delta)\bigr)+Q_\varepsilon(A_\delta^{c})\ln\!\bigl(1+1/Q_\varepsilon(A_\delta^{c})\bigr)$,
which also tends to zero. Hence both KL terms tend to $\ln 2$, and since
$\mathrm{JSD}=\tfrac{1}{2}\mathrm{KL}(P\|M)+\tfrac{1}{2}\mathrm{KL}(Q\|M)$, the limit is $\ln 2$.
\end{proof}
\noindent The bound $\ln 2$ is attained by any flow whose backward direction carries no
variation, since a point mass against a spread-out forward distribution is the limit in (ii).
It therefore marks unidirectionality, which some attack families exhibit almost without
exception and which benign single-direction flows exhibit as well; Section~\ref{sec:framework}
reports how often each class and family reaches it.

\begin{figure}[t]
  \centering
  \includegraphics[width=\linewidth]{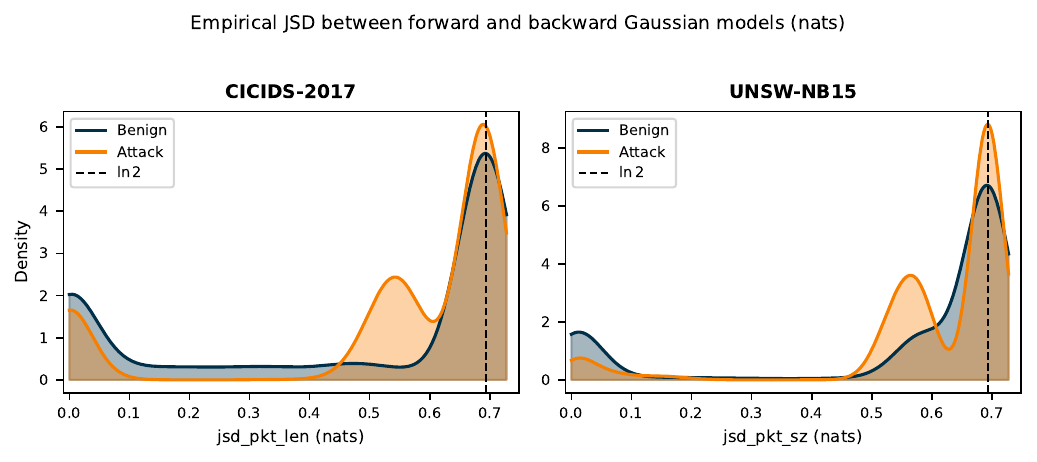}
  \caption{Empirical distributions of the cross-directional JSD for attack and benign flows (kernel density estimate, KDE) in CICIDS-2017 (\texttt{jsd\_pkt\_len}) and UNSW-NB15 (\texttt{jsd\_pkt\_sz}). Both classes are bimodal, with mass near zero (symmetric flows) and at the bound $\ln 2\approx0.693$ nats (dashed; one direction without variation), and the two classes differ in the balance between the modes rather than in their location.}
  \label{fig:jsd_empirical}
\end{figure}

Fig.~\ref{fig:jsd_empirical} shows that the JSD is bimodal for both classes in both datasets, with one mode near zero and one at the bound. Counting a flow as at the bound when its value is within 0.001 nats of $\ln 2$, the bound is reached by 40.5\% of attack flows and 57.8\% of benign flows in CICIDS-2017, and by 57.7\% and 60.2\% in UNSW-NB15, so proximity to $\ln 2$ is not a signature of attacks in general. It is instead a signature of unidirectional families. In CICIDS-2017, 99.6\% of PortScan flows sit at the bound and DDoS flows lie at or just below it (36.9\% at the bound, median 0.681), against 4.5\% of DoS Hulk and 26.1\% of Bot flows, as Proposition~\ref{prop:jsd} predicts for one-directional bursts. The classes differ instead in the low mode. Attack flows are less often symmetric than benign flows (JSD below 0.05 in 15.4\% against 21.7\% of flows in CICIDS-2017 and 4.9\% against 16.0\% in UNSW-NB15), and CICIDS-2018 reverses the picture, with attacks more symmetric than benign traffic (median 0.09 against 0.31). The feature carries class information through the shape of this distribution, which a tree ensemble can exploit, not through a single threshold near the bound.

\begin{figure*}[t]
\centering
\includegraphics[width=\textwidth]{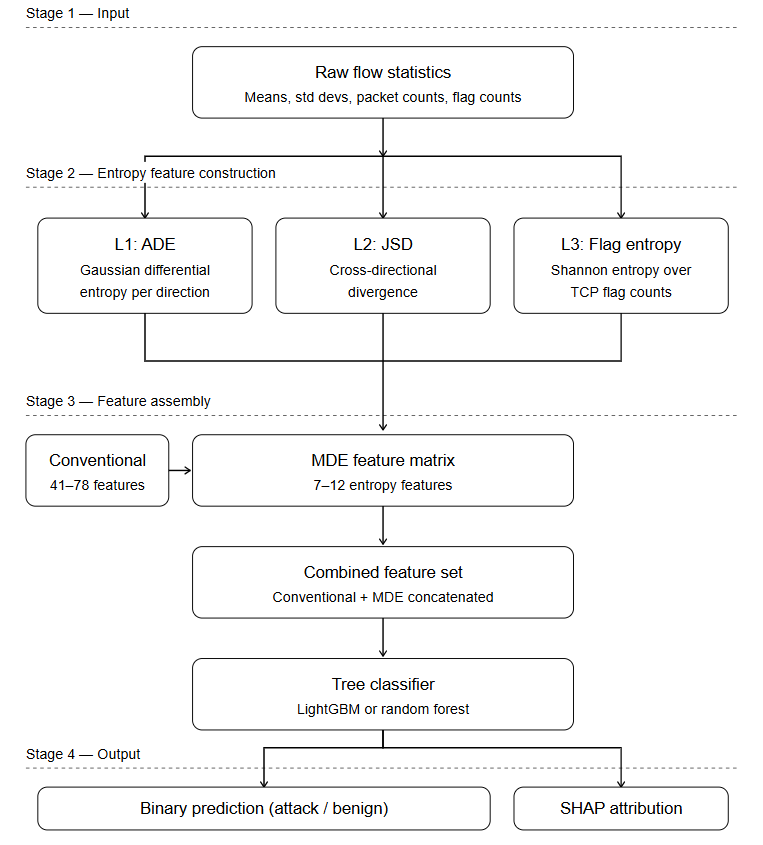}
\caption{The MDE pipeline. Pre-aggregated flow statistics feed three analytical entropy levels (L1 within-flow ADE, L2 cross-directional JSD, L3 flag-incidence entropy), whose outputs are concatenated with the conventional features and passed to a tree ensemble with per-instance SHAP; the entropy-only and conventional ablations use each branch alone. The levels are defined in Section~\ref{sec:framework}.}
\label{fig:architecture}
\end{figure*}

\begin{figure}[t]
  \centering
  \includegraphics[width=\linewidth]{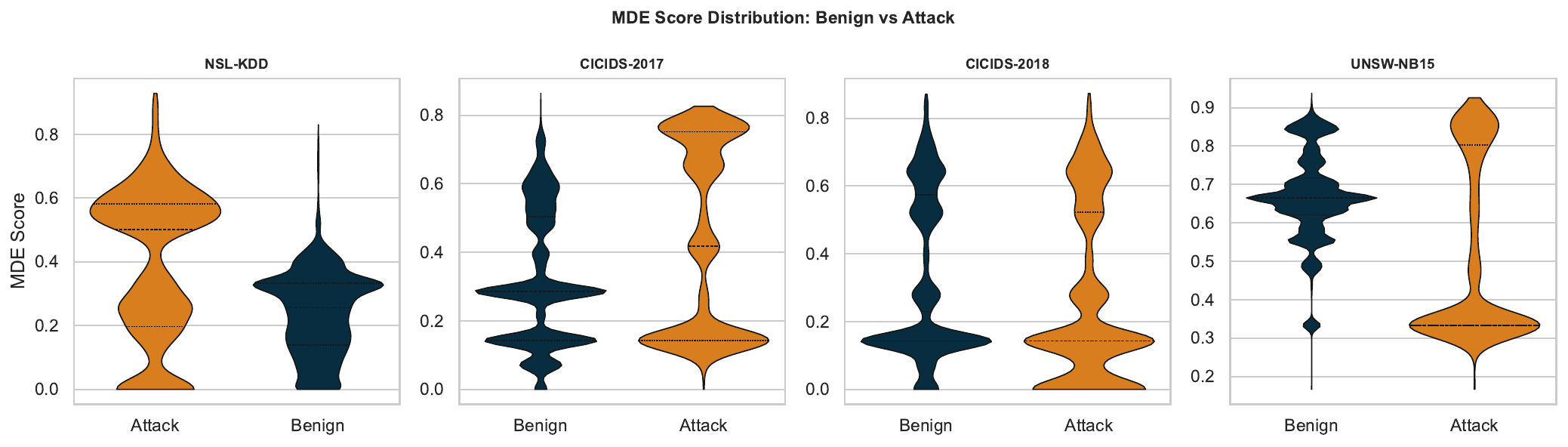}
  \caption{MDE score distributions (Benign vs.\ Attack) across all four benchmark datasets. The degree of class separation varies by dataset, reflecting each dataset's attack structure.}
  \label{fig:framework_overview}
\end{figure}

\subsection{Entropy Levels}

\paragraph{Level 1 (ADE).} For a network flow with forward packet length mean $\mu_f$ and standard deviation $\sigma_f$, we approximate the within-flow packet size distribution as Gaussian and compute:
\begin{equation}
  h^{\text{ADE}}_{\text{fwd}} = \tfrac{1}{2}\ln(2\pi e\,\sigma_f^2).
  \label{eq:ade}
\end{equation}
The same is computed for the backward direction ($h^{\text{ADE}}_{\text{bwd}}$) and for inter-arrival times. Where minimum and maximum packet lengths are available, a complementary uniform range entropy $h^{\text{range}} = \ln(\text{max}-\text{min})$ is also computed. High differential entropy in packet sizes indicates variable, structurally complex traffic; attack flows such as scanning or DDoS typically exhibit low variance in packet size (all probes or all flood packets are identical), producing low ADE values.

\paragraph{Level 2 (JSD).} A structural asymmetry between forward (source-to-destination) and backward (destination-to-source) traffic characterizes many attack types. DDoS floods are almost entirely unidirectional, port scans generate probes in one direction with negligible responses, and legitimate sessions tend to be bidirectional and more balanced. MDE captures this via the JSD between the approximate forward and backward packet-length distributions:
\begin{equation}
  \mathrm{JSD}_{\text{pkt}} = \mathrm{JSD}\!\left(\mathcal{N}(\mu_f,\sigma_f^2)\,\|\,\mathcal{N}(\mu_b,\sigma_b^2)\right).
  \label{eq:jsd}
\end{equation}
The JSD between two Gaussians has no closed form, so it is computed numerically. It depends on the
two densities only through the standardized separation $\delta = |\mu_b-\mu_f|/\sigma_{\max}$ and the
scale ratio $\rho = \sigma_{\min}/\sigma_{\max}$, and with $w = p/(p+q)$ and $H_b$ the binary entropy
it equals $\ln 2 - \tfrac{1}{2}\mathbb{E}_P[H_b(w)] - \tfrac{1}{2}\mathbb{E}_Q[H_b(w)]$. Each
expectation is evaluated by composite 16-point Gauss--Legendre quadrature on panels whose breakpoints
follow the scales of both Gaussians, which resolves a narrow direction lying inside a wide one, and
the pipeline reads the resulting function of $(\ln\rho,\delta)$ from a cubic-spline table. Against
adaptive integration, the composite rule has a maximum absolute error below $10^{-12}$ nats and the
tabulated form below $10^{-6}$ nats, both on random parameters and on flow parameters sampled from
every dataset. A single Gauss--Hermite rule under each Gaussian does not resolve the narrow
component and errs by more than $0.01$ nats in up to a quarter of the flows, so it is not used.
The value is the divergence of the two fitted densities themselves and lies in $[0,\ln 2]$, so
Proposition~\ref{prop:jsd} applies to the computed quantity. A standard deviation reported as zero is
floored at $10^{-9}$, so a direction with no variation behaves as a point mass and its divergence from
a spread-out opposite direction is $\ln 2$, the limit in Proposition~\ref{prop:jsd}(ii).
A directional balance entropy $H_{\text{dir}} = -r\log r - (1-r)\log(1-r)$, where $r \in [0,1]$ is the fraction of packets (or bytes) in the forward direction, provides a complementary measure of traffic symmetry independent of the Gaussian approximation.

\paragraph{Level 3 (flag-incidence entropy).} TCP control flags (FIN, SYN, RST, PSH, ACK, URG) carry protocol-level intent. Benign connections raise a varied set of flags over their lifetime; attacks often manipulate a narrow subset (e.g.\ SYN-only for SYN floods, RST-heavy for reset injections). Given the per-flow flag counts $c_k$ for $k \in \{\text{FIN,SYN,RST,PSH,ACK,URG}\}$ and total flag count $C = \sum_k c_k$, the flag-incidence entropy is:
\begin{equation}
  H_{\text{flags}} = -\sum_{k} \frac{c_k}{C}\log\frac{c_k}{C}.
  \label{eq:flag}
\end{equation}
Here $c_k/C$ is the share of flag type $k$ among all flag incidences in the flow. Because a packet can carry several flags at once (PSH together with ACK, for instance), the counts are incidences rather than mutually exclusive packet states, and $H_{\text{flags}}$ measures the diversity of flag incidences within a flow, not the entropy of a flag sequence. $H_{\text{flags}}=0$ when a single flag accounts for every incidence ($c_k = C$ for one $k$) and $H_{\text{flags}} = \ln 6$ when all six flag types occur equally often. Low values identify flows dominated by one flag type, as in SYN flooding or reset injection.

\subsection{Composite MDE Score}

A composite score aggregates across levels by normalizing each entropy feature to $[0,1]$ using training-set statistics and averaging:
\begin{equation}
  s_{\text{MDE}} = \frac{1}{|\mathcal{F}|}\sum_{f\in\mathcal{F}}\frac{h_f - \min_{\mathcal{T}} h_f}{\max_{\mathcal{T}} h_f - \min_{\mathcal{T}} h_f},
  \label{eq:mde}
\end{equation}
where $\mathcal{F}$ is the set of entropy features available for the given dataset schema, $h_f$ is the value of entropy feature $f$ for the current flow, and $\min_{\mathcal{T}} h_f$, $\max_{\mathcal{T}} h_f$ are the minimum and maximum of feature $f$ over the training set $\mathcal{T}$. For temporal and hold-out experiments, normalization statistics are computed on the historical training corpus and applied unchanged to held-out flows. For cross-validation experiments, LightGBM and Random Forest are invariant to monotonic feature rescaling, so the normalization choice does not affect any reported metric; $s_\text{MDE}$ serves as an interpretability summary in that context. This score provides a single interpretable summary of the overall distributional complexity of a flow, ranging from 0 (maximally regular) to 1 (maximally entropic relative to the training distribution).

\section{Experimental Setup}\label{sec:setup}
\vspace{0.5em}

This section fixes the evaluation protocol. It describes the four benchmarks and the structural diversity they span, how the three entropy levels are instantiated on each schema, the two tree-ensemble classifiers and the full metric suite, and the fold-local measures that keep every preprocessing and feature statistic on the training side of each split.

\subsection{Datasets}\label{sec:datasets}

Four publicly available benchmark datasets are selected to span three structural dimensions: (i)~traffic generation era (1999--2018), (ii)~attack taxonomy breadth (binary to 12 categories), and (iii)~feature schema heterogeneity (NSL-KDD connection records to CICFlowMeter packet statistics). Class imbalance ranges from near-balanced (NSL-KDD 43/57) to moderate (CICIDS-2018 72/28) and heavy (UNSW-NB15 93/7). This diversity is necessary to test whether the analytical MDE framework, which adapts its entropy computations to each schema, produces consistent, interpretable signals across structurally distinct environments. Table~\ref{tab:datasets_summary} provides a structured summary; key properties are elaborated below.

\begin{table}[t]
\centering
\caption{Dataset summary. \#Feat: number of numerical features after cleaning. Imbal.: benign\,:\,attack ratio (approximate).}
\label{tab:datasets_summary}
\footnotesize
\setlength{\tabcolsep}{3pt}
\resizebox{\linewidth}{!}{%
\begin{tabular}{l r r r r}
\toprule
\textbf{Dataset} & \textbf{Flows} & \textbf{\#Feat} & \textbf{Atk types} & \textbf{Imbal.} \\
\midrule
NSL-KDD     &  22,544 & 41 &  4 & 43:57 \\
CICIDS-2017 & 250,000 & 78 & 12 & 80:20 \\
CICIDS-2018 & 150,000 & 78 &  2 & 72:28 \\
UNSW-NB15   & 200,000 & 41 &  9 & 93:7  \\
\bottomrule
\end{tabular}}
\end{table}

NSL-KDD~\cite{tavallaee2009kdd} is selected as the canonical baseline for IDS research. NSL-KDD is a curated refinement of KDD Cup 1999 that removes duplicate records present in the original dataset, yielding 22,544 flows with 41 connection-level features covering SYN flood, probing, R2L, and U2R attack categories. Its balanced class distribution (43\% benign, 57\% attack) and well-characterized properties make it the standard reproducibility anchor against which new approaches are benchmarked. \emph{Known artifact}: the KDD feature set includes high-level connection summaries (e.g., \texttt{serror\_rate}, \texttt{same\_srv\_rate}) that encode attack context implicitly, contributing to high conventional-feature separability.

CICIDS-2017~\cite{sharafaldin2018cicids} is selected for its multi-attack breadth and temporal structure. Generated over five days by the Canadian Institute for Cybersecurity using real traffic replayed through a controlled network, CICIDS-2017 contains 12 attack categories including DDoS, port scanning, brute-force, web attacks, infiltration, and Heartbleed, covering the broadest attack taxonomy of the four evaluated benchmarks. The temporal day-by-day structure enables the time-based generalization experiment of Section~\ref{sec:timesplit}. A 250K stratified sample is used for CV experiments. \emph{Known artifact}: Engelen et al.~\cite{engelen2021troubleshooting} document labeling inconsistencies and near-constant features for certain attack classes; some categories are separable by a single CICFlowMeter statistic (e.g., \texttt{Flow Duration}$=0$), which contributes to near-saturated conventional-feature F1.

CICIDS-2018~\cite{sharafaldin2018cicids} is selected as a harder contemporary counterpart to CICIDS-2017. Generated independently in 2018 under a partially overlapping attack taxonomy, CICIDS-2018 exhibits substantially higher intra-class distributional overlap than CICIDS-2017, yielding a genuine classification challenge that stress-tests the entropy framework under conditions closer to operational complexity. A 150K stratified sample is used. Two properties of the distributed file require care. It repeats its header line inside the data (25 such rows in the full file), which leaves every numeric column typed as text in a naive read; these rows are removed and the columns are coerced to numbers before any processing. It also uses pre-encoded binary labels (\texttt{class}=0/1), and the original label column must be excluded from the features to prevent trivial leakage (Section~\ref{sec:leakage}). The destination port, the only identifier column present in this file, is removed with the timestamp.

UNSW-NB15~\cite{moustafa2015unsw} is selected for its modern attack categories and rich feature schema. Collected at the Cyber Range Lab of UNSW Canberra using the IXIA PerfectStorm tool, UNSW-NB15 contains 9 attack types (Fuzzers, DoS, Exploits, Generic, Reconnaissance, Shellcode, Worms, Backdoors, Analysis) against realistic background traffic. Its 41 features include network-level statistics not present in CICFlowMeter datasets: jitter, connection load, time-to-live (TTL) asymmetry, and service type, which provide richer proxies for L1 ADE and L3 entropy computation. A 200K stratified sample is used. \emph{Known artifact}: the source TTL alone separates the two classes with an AUC of 0.989 in the sample used here, a property of the testbed that every feature derived from TTL inherits, and the class imbalance (93\% benign, 7\% attack) is addressed via balanced class weights.

Fig.~\ref{fig:class_dist} shows the class distributions. All datasets are partitioned using stratified sampling to maintain class proportions across training, validation, and test folds.

\begin{figure}[t]
  \centering
  \includegraphics[width=\linewidth]{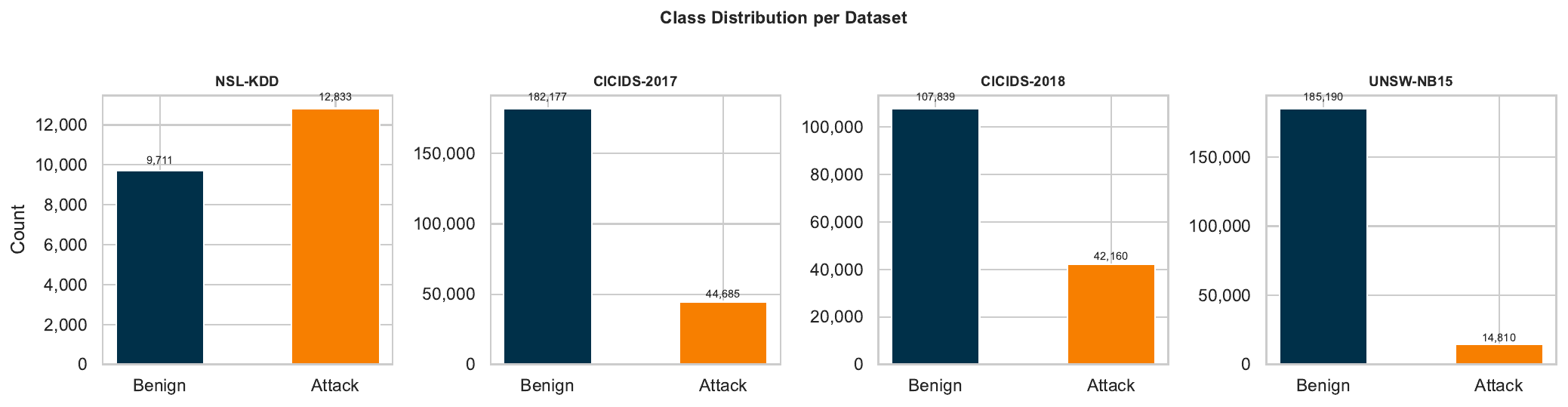}
  \caption{Class distribution (Benign vs.\ Attack) across all four benchmark datasets after stratified sampling.}
  \label{fig:class_dist}
\end{figure}

\subsection{Feature Construction and Classifiers}

Because the four benchmarks use different feature schemas, the three levels are instantiated per dataset, and the instantiations are not mathematically identical (Table~\ref{tab:inventory}). On CICIDS-2017/2018 every level is computed from the quantity it is defined on. The Gaussian ADE uses the per-direction standard deviations of packet length and inter-arrival time (IAT), the uniform-range entropy uses the per-direction minimum and maximum packet length, the JSD uses the per-direction means and standard deviations, and the flag-incidence entropy uses the TCP flag counts. UNSW-NB15 records neither per-direction standard deviations of packet size nor flag counts, so jitter serves as the scale parameter of both the IAT and the packet-size Gaussians, the range entropy takes the mean packet size as the range, and the logarithm of the TTL difference between the two directions stands in for L3. NSL-KDD has no directional moments. Its rate features (serror\_rate, rerror\_rate, same\_srv\_rate, and their service-level counterparts) are normalized into a probability vector whose Shannon entropy is the L1 signal, service diversity and byte direction enter as binary entropies, and the byte-asymmetry JSD is taken between two Gaussians centered at the log byte counts with an assumed relative scale of 10\%. Each schema also contributes one or two log-scaled volume descriptors, which are not information-theoretic quantities and whose contribution is measured separately in Section~\ref{sec:simplestats}. The schemas therefore share the three-level design and the definitions of the entropy functionals, while the inputs, and on UNSW-NB15 and NSL-KDD some of the quantities themselves, are proxies.

\begin{table}[!htb]
\centering
\caption{MDE feature inventory per schema (counts in parentheses). Where a schema lacks a defining input, the named proxy is used (Section~\ref{sec:setup}); the log-scaled descriptors are not entropy terms.}
\label{tab:inventory}
\scriptsize
\setlength{\tabcolsep}{2.5pt}
\begin{tabular}{@{}p{1.15cm}p{2.25cm}p{2.25cm}p{2.25cm}@{}}
\toprule
 & \textbf{CICIDS-2017/2018} & \textbf{UNSW-NB15} & \textbf{NSL-KDD} \\
\midrule
L1 & Gaussian ADE of packet length and IAT per direction (4); uniform-range entropy of packet length per direction (2) & Gaussian ADE per direction with jitter as scale (2); range entropy with mean packet size as range (1) & Shannon entropy of the rate vector (1); binary entropy of service diversity (1) \\
\midrule
L2 & Gaussian JSD of packet length and of IAT (2); direction entropy of packets (1) & Gaussian JSD of packet size and of IAT with jitter as scale (2); direction entropy of bytes and packets (2) & Direction entropy of bytes (1); Gaussian JSD of log byte counts with assumed 10\% scale (1) \\
\midrule
L3 & Flag-incidence entropy (1) & Log TTL difference (1) & None \\
\midrule
Composite & MDE score (1) & MDE score (1) & MDE score (1) \\
\midrule
Log-scaled & Byte rate (1) & Source and destination load (2) & Source and destination bytes (2) \\
\midrule
Total & 12 & 11 & 7 \\
\bottomrule
\end{tabular}
\end{table}

For each dataset, three feature matrices are constructed to isolate the contribution of MDE features:
\begin{enumerate}[nosep,label=(\roman*)]
  \item Conventional: original numerical flow features only (no entropy).
  \item Entropy-only: MDE features only (7--12 depending on the schema, including their log-scaled descriptors; Table~\ref{tab:inventory}).
  \item Combined: concatenation of conventional and MDE features.
\end{enumerate}
This three-way ablation isolates the marginal contribution of entropy features and establishes whether they can serve as standalone representations. The experiments rely on two tree-ensemble classifiers, LightGBM~\cite{ke2017lightgbm} and Random Forest~\cite{breiman2001random}. LightGBM is configured with 300 estimators, a learning rate of 0.05, 63 leaves, and balanced class weights, while Random Forest uses 200 estimators, a maximum depth of 20, and balanced class weights. Both are trained under stratified 5-fold cross-validation using a fold-local preprocessing pipeline. Within each fold, a \texttt{SimpleImputer}(strategy=median) is fitted on the training split and applied to both training and test splits, followed by a \texttt{PercentileClipper} (99.9th percentile) fitted and applied in the same fold-local manner. This prevents any cross-fold leakage from preprocessing statistics. The entropy features are constructed inside the same pipeline by a transformer whose imputation medians and composite-score bounds are fitted on the training split of each fold (Section~\ref{sec:leakage}), so no preprocessing statistic crosses a fold boundary. Class imbalance is handled through the \texttt{class\_weight=`balanced'} setting in both classifiers, and infinite values arising from division by zero in flow feature computation are treated as missing and imputed in the same fold-local manner.

\subsection{Evaluation Metrics}\label{sec:metrics}

Let $\mathrm{TP}$, $\mathrm{TN}$, $\mathrm{FP}$, and $\mathrm{FN}$ denote true positives, true negatives, false positives, and false negatives at the binary (attack/benign) decision boundary. We report the following metrics, selected to reflect both detection capability and operational cost in IDS deployment:

\begin{flalign}
& \text{Precision} = \frac{\mathrm{TP}}{\mathrm{TP}+\mathrm{FP}} & \label{eq:prec}\\[2pt]
& \text{Recall (DR)} = \frac{\mathrm{TP}}{\mathrm{TP}+\mathrm{FN}} & \label{eq:recall}\\[2pt]
& \text{FAR (FPR)} = \frac{\mathrm{FP}}{\mathrm{FP}+\mathrm{TN}} & \label{eq:far}\\[2pt]
& F_1 = \frac{2\cdot\text{Precision}\cdot\text{Recall}}{\text{Precision}+\text{Recall}} & \label{eq:f1}\\[2pt]
& \text{MCC} = \frac{\mathrm{TP}\cdot\mathrm{TN}-\mathrm{FP}\cdot\mathrm{FN}}{\sqrt{(\mathrm{TP}+\mathrm{FP})(\mathrm{TP}+\mathrm{FN})(\mathrm{TN}+\mathrm{FP})(\mathrm{TN}+\mathrm{FN})}} & \label{eq:mcc}
\end{flalign}

Each metric serves a distinct diagnostic purpose. \emph{Precision} (Eq.~\ref{eq:prec}) measures alert reliability. Low precision means many false alarms that burden analyst triage.
\emph{Recall/DR} (Eq.~\ref{eq:recall}) measures attack coverage. Low recall means missed intrusions with potentially severe consequences.
\emph{FAR/FPR} (Eq.~\ref{eq:far}) quantifies the false-alarm operational burden directly relevant to security operations center (SOC) workload~\cite{mahbooba2021explainable}.
\emph{$F_1$} (Eq.~\ref{eq:f1}) balances precision and recall, making it the primary aggregate metric for imbalanced IDS datasets~\cite{khraisat2019survey}. Weighted $F_1$ accounts for class-size differences across the multi-class setting.
\emph{MCC} (Eq.~\ref{eq:mcc}) is a correlation coefficient that provides a single balanced summary even under extreme class imbalance~\cite{chicco2020advantages}; it is reported for multi-class aggregation where weighted $F_1$ may be dominated by the majority class.
\emph{ROC-AUC} is the area under the receiver operating characteristic curve, measuring rank discrimination independent of a classification threshold~\cite{fawcett2006introduction}. AUC is particularly informative under distribution shift (Section~\ref{sec:timesplit}), where per-threshold metrics conflate threshold selection with model quality.
\emph{PR-AUC} (area under the precision-recall curve) is the preferred discrimination metric under heavy class imbalance~\cite{saito2015precision}, reported in Tables~\ref{tab:fullmetrics}, \ref{tab:timesplit}, and~\ref{tab:inpipeline}.

Every experiment in this paper records the same complete suite for every evaluation set: the confusion-matrix counts (TP, FN, FP, TN), accuracy, weighted and macro-averaged precision, recall, and $F_1$, the attack-class precision, recall, and $F_1$, DR, FAR, the false negative and true negative rates, balanced accuracy, MCC, ROC-AUC, and PR-AUC. Under cross-validation the suite is computed per fold, averaged with its standard deviation, and the counts are pooled; the per-fold values are part of the released result tables. The printed tables show the subset that each comparison turns on: Table~\ref{tab:ablation} gives weighted $F_1$, Precision, and Recall for each condition and classifier, Table~\ref{tab:fullmetrics} gives the full suite with pooled confusion counts for the combined condition, and Section~\ref{sec:perclass} gives the per-class breakdown.

\subsection{Leakage Prevention}\label{sec:leakage}

Reproducible IDS evaluation requires deliberate protocol choices to prevent data contamination. We address four sources of leakage. First, preprocessing (median imputation, 99.9th-percentile clipping) is applied \emph{fold-locally} inside a \texttt{Pipeline}. The imputer and clipper are fitted only on each training fold and then applied to the corresponding validation fold, so no global statistics cross the fold boundary. Second, MDE features are deterministic functions of each flow's own statistics (Eqs.~\ref{eq:ade}--\ref{eq:flag}) with no label dependence, but two of their ingredients are statistics of the data: the medians that replace missing or infinite inputs and the min-max bounds that scale the composite score. Both are therefore fitted on the training fold and applied to the validation fold, as a transformer step inside the same \texttt{Pipeline}, so that no statistic of the validation flows enters the features of the training flows. Third, all identifier columns (IP addresses, ports, timestamps, flow IDs) are removed before feature matrix construction. Fourth, the original label column is used only to derive binary targets and is then dropped, ensuring it cannot enter any feature set in encoded form. Without this correction, the encoded label would make several CICIDS-2018 attack categories trivially separable and inflate the measured scores, masking the genuine distributional overlap that makes the dataset a meaningful stress test.

Campaign-level correlation remains an important caveat. Stratified random CV allows flows from the same attack campaign to appear in both train and test folds, inflating performance relative to real deployment. The temporal split (Section~\ref{sec:timesplit}) is designed to quantify this optimism on CICIDS-2017. Host-based and session-level isolation would offer stricter guarantees and are recommended as evaluation standards for future work.

\section{Results}\label{sec:results}
\vspace{0.5em}

The evaluation proceeds from within-distribution behavior to progressively harder generalization. The ablation and full-metric reporting first establish what the entropy features contribute under cross-validation, after which temporal, pseudo-live, cross-dataset, and unseen-family settings probe how the same models behave under distribution shift.

\subsection{Ablation Study}

Table~\ref{tab:ablation} reports weighted F1, Precision, and Recall for each feature condition; Table~\ref{tab:fullmetrics} adds DR, FAR, Acc, MCC, AUC, PR-AUC, and the pooled confusion counts for the combined condition. On CICIDS-2018, combined F1$\approx$0.73 masks a DR of only 0.42 to 0.44 with FAR between 0.12 and 0.14, meaning fewer than half of the attacks are detected; none of this is visible from the aggregate score alone. On CICIDS-2017 entropy-only, DR$=$0.986 comes at FAR$=$0.042 under LightGBM and DR$=$0.984 at FAR$=$0.050 under Random Forest, flagging 4.2\% to 5.0\% of benign flows, a precision-recall trade-off that the aggregate F1 absorbs without trace. These cases motivate the full metric disclosure used throughout.

\begin{table*}[t]
\centering
\caption{Ablation results: weighted $F_1$, Precision (P), and Recall (R) for each feature condition, 5-fold stratified CV. Bold: best $F_1$ per dataset. $n_f$: feature counts (conv/ent/comb); on CICIDS-2018 the conventional count is 77 after the encoded label column is dropped (Section~\ref{sec:leakage}). Full 9-metric results (DR, FAR, Acc, MCC, AUC, PR-AUC) for the combined condition appear in Table~\ref{tab:fullmetrics}.}
\label{tab:ablation}
\footnotesize
\setlength{\tabcolsep}{3pt}
\resizebox{\textwidth}{!}{%
\begin{tabular}{ll r  rrr  rrr  rrr}
\toprule
 & & & \multicolumn{3}{c}{\textbf{Conventional}} & \multicolumn{3}{c}{\textbf{Entropy-only}} & \multicolumn{3}{c}{\textbf{Combined}} \\
\cmidrule(lr){4-6}\cmidrule(lr){7-9}\cmidrule(lr){10-12}
\textbf{Dataset} & \textbf{Model} & $n_f$ & F1 & P & R & F1 & P & R & F1 & P & R \\
\midrule
NSL-KDD     & LightGBM     & 41/7/48  & \textbf{0.9873} & 0.9873 & 0.9873 & 0.9783 & 0.9784 & 0.9783 & 0.9871 & 0.9871 & 0.9871 \\
            & RandomForest &          & 0.9852 & 0.9852 & 0.9852 & 0.9729 & 0.9729 & 0.9729 & \textbf{0.9855} & 0.9855 & 0.9855 \\
\midrule
CICIDS-2017 & LightGBM     & 78/12/90 & \textbf{0.9989} & 0.9989 & 0.9989 & 0.9648 & 0.9682 & 0.9638 & \textbf{0.9989} & 0.9989 & 0.9989 \\
            & RandomForest &          & \textbf{0.9977} & 0.9977 & 0.9977 & 0.9582 & 0.9630 & 0.9569 & \textbf{0.9977} & 0.9977 & 0.9977 \\
\midrule
CICIDS-2018 & LightGBM     & 77/12/89 & 0.7311 & 0.7267 & 0.7409 & 0.7191 & 0.7146 & 0.7346 & \textbf{0.7313} & 0.7269 & 0.7409 \\
            & RandomForest &          & 0.7339 & 0.7301 & 0.7466 & 0.7219 & 0.7197 & 0.7416 & \textbf{0.7348} & 0.7312 & 0.7480 \\
\midrule
UNSW-NB15   & LightGBM     & 41/11/52 & \textbf{0.9926} & 0.9930 & 0.9925 & 0.9887 & 0.9897 & 0.9884 & \textbf{0.9926} & 0.9929 & 0.9925 \\
            & RandomForest &          & 0.9898 & 0.9907 & 0.9895 & 0.9877 & 0.9888 & 0.9872 & \textbf{0.9900} & 0.9909 & 0.9897 \\
\bottomrule
\end{tabular}%
}
\end{table*}

\begin{table*}[t]
\centering
\caption{Full operational metrics for the \textbf{combined} feature set (5-fold stratified CV, fold-local pipeline). TP, FN, FP, TN are confusion-matrix counts pooled over the five folds. Bold: best F1 per dataset. CICIDS-2018 is the critical case: F1$\approx$0.73 conceals DR$\approx$0.44, meaning fewer than half of all attacks are detected.}
\label{tab:fullmetrics}
\footnotesize
\setlength{\tabcolsep}{3pt}
\begin{tabular}{ll rrrrrrrrr rrrr}
\toprule
\textbf{Dataset} & \textbf{Model} & F1 & Prec & Rec & DR & FAR & Acc & MCC & AUC & PR-AUC & TP & FN & FP & TN \\
\midrule
NSL-KDD     & LightGBM     & \textbf{0.9871} & 0.9871 & 0.9871 & 0.9885 & 0.0147 & 0.9871 & 0.9737 & 0.9995 & 0.9997 & 12,685 & 148 & 143 & 9,568 \\
            & RandomForest & 0.9855 & 0.9855 & 0.9855 & 0.9839 & 0.0126 & 0.9855 & 0.9704 & 0.9994 & 0.9995 & 12,627 & 206 & 122 & 9,589 \\
\midrule
CICIDS-2017 & LightGBM     & \textbf{0.9989} & 0.9989 & 0.9989 & 0.9983 & 0.0010 & 0.9989 & 0.9964 & 0.9999 & 0.9998 & 44,609 & 76 & 184 & 181,993 \\
            & RandomForest & 0.9977 & 0.9977 & 0.9977 & 0.9952 & 0.0017 & 0.9977 & 0.9927 & 0.9998 & 0.9993 & 44,471 & 214 & 307 & 181,870 \\
\midrule
CICIDS-2018 & LightGBM     & 0.7313 & 0.7269 & 0.7409 & 0.4381 & 0.1407 & 0.7409 & 0.3206 & 0.7089 & 0.5830 & 18,469 & 23,691 & 15,168 & 92,671 \\
            & RandomForest & \textbf{0.7348} & 0.7312 & 0.7480 & 0.4175 & 0.1228 & 0.7480 & 0.3278 & 0.7065 & 0.5794 & 17,600 & 24,560 & 13,238 & 94,601 \\
\midrule
UNSW-NB15   & LightGBM     & \textbf{0.9926} & 0.9929 & 0.9925 & 0.9892 & 0.0073 & 0.9925 & 0.9478 & 0.9996 & 0.9950 & 14,650 & 160 & 1,347 & 183,843 \\
            & RandomForest & 0.9900 & 0.9909 & 0.9897 & 0.9973 & 0.0109 & 0.9897 & 0.9316 & 0.9995 & 0.9933 & 14,770 & 40 & 2,014 & 183,176 \\
\bottomrule
\end{tabular}
\end{table*}

\begin{figure*}[t]
  \centering
  \includegraphics[width=\textwidth]{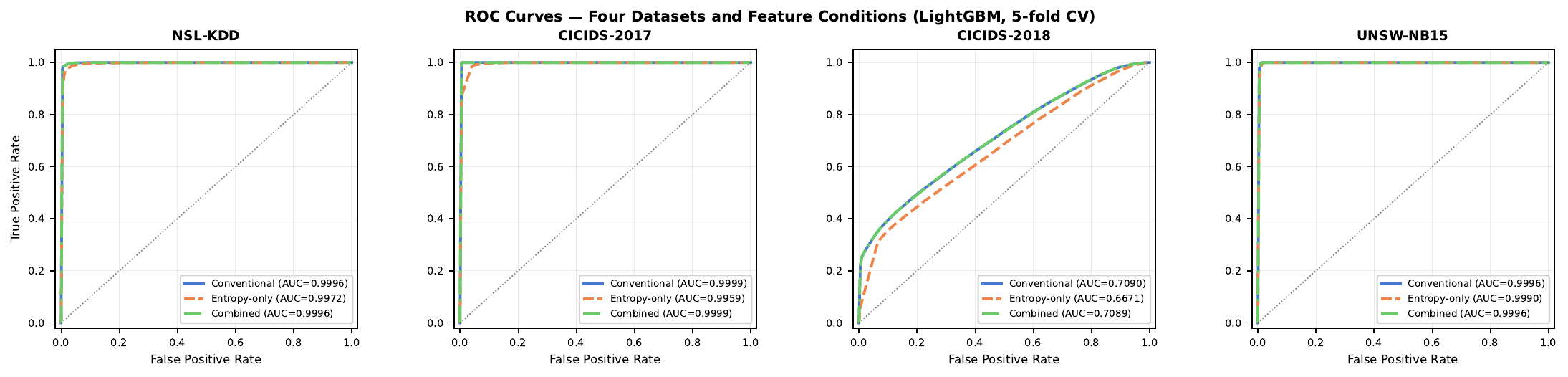}
  \caption{ROC curves for all four datasets across three feature conditions (LightGBM, 5-fold CV). Solid blue: conventional; dashed amber: entropy-only; dash-dot green: combined. AUC values are means over folds. CICIDS-2018 shows the lowest AUC ($\approx$0.73), consistent with genuine feature-space ambiguity between attack categories. Combined condition matches or exceeds conventional on all datasets, so entropy features do not degrade discriminative ranking.}
  \label{fig:roc_all}
\end{figure*}

To quantify whether observed F1 differences are statistically reliable, we compute 95\% confidence intervals from the 5-fold CV scores using the $t$-distribution ($t_{4,0.025}=2.776$) and apply Wilcoxon signed-rank tests~\cite{wilcoxon1945} for paired comparisons. With $k=5$ folds, the minimum achievable Wilcoxon $p$-value is $1/2^4=0.0625$; we adopt a one-sided significance criterion of $p\leq0.0625$ for this sample size.

Combining MDE entropy with conventional features changes F1 by at most 0.09 percentage points, and the change is not statistically significant in seven of the eight dataset-model pairs (Wilcoxon signed-rank over the five folds, $p \geq 0.25$; confidence intervals overlap). Under Random Forest on CICIDS-2018 the combined set is higher in all five folds ($p=0.0625$, the minimum attainable with five folds), by 0.09 percentage points, a difference too small to matter operationally. 95\% CIs for the conventional and combined conditions under LightGBM are NSL-KDD $[0.985, 0.990]$ / $[0.985, 0.990]$; CICIDS-2017 $[0.9985, 0.9993]$ / $[0.9986, 0.9991]$; CICIDS-2018 $[0.728, 0.734]$ / $[0.729, 0.734]$; UNSW-NB15 $[0.992, 0.993]$ / $[0.992, 0.993]$. Entropy features do not degrade conventional performance.

In contrast, the entropy-only condition achieves lower F1 than conventional on all four datasets in every fold (Wilcoxon $p=0.0625$ per dataset, the minimum attainable with five folds, and $p=0.0001$ over the twenty dataset-folds pooled, for both classifiers). 95\% CIs under LightGBM: NSL-KDD entropy $[0.976, 0.980]$ vs.\ conventional $[0.985, 0.990]$; CICIDS-2017 $[0.964, 0.966]$ vs.\ $[0.9985, 0.9993]$; CICIDS-2018 $[0.716, 0.722]$ vs.\ $[0.728, 0.734]$; UNSW-NB15 $[0.988, 0.989]$ vs.\ $[0.992, 0.993]$. Across both classifiers the gap is smallest on UNSW-NB15 (0.2--0.4 percentage points) and largest on CICIDS-2017 (3.4--4.0 pp), consistent with the known separability of CICIDS-2017 by packet-volume features.

\subsection{Entropy-Only Results}

The significance tests establish that entropy features do not degrade F1 relative to conventional features and raise it, where they do, by less than 0.1 percentage points. The entropy-only condition, though weaker in absolute F1, is more revealing. Its per-dataset variation, shown in Fig.~\ref{fig:ablation}, exposes where the MDE signal is strong and where it is limited by dataset structure.

UNSW-NB15 yields the strongest entropy-only performance (F1 $=$ 0.9887/0.9877, DR $=$ 0.989/0.991 for LightGBM/Random Forest), on a dataset whose source TTL alone separates the classes with an AUC of 0.989 (Section~\ref{sec:datasets}), so the result owes as much to the dataset as to the representation. NSL-KDD achieves F1 $=$ 0.9783/0.9729 (DR $=$ 0.977/0.975), where the connection-state entropy captures the structured probe-and-scan signatures characteristic of KDD-style attacks. CICIDS-2018 is the weakest case, with entropy-only F1 $=$ 0.719/0.722 and DR $=$ 0.36--0.38. The combined condition reaches F1 $=$ 0.731/0.735, and every classifier in Table~\ref{tab:inpipeline} detects fewer than half of the attacks. The results are consistent with substantial feature-space ambiguity in this dataset rather than a failure unique to MDE, although the evaluation cannot exclude that a different model family or feature set would separate the classes better.

\begin{figure}[t]
  \centering
  \includegraphics[width=\linewidth]{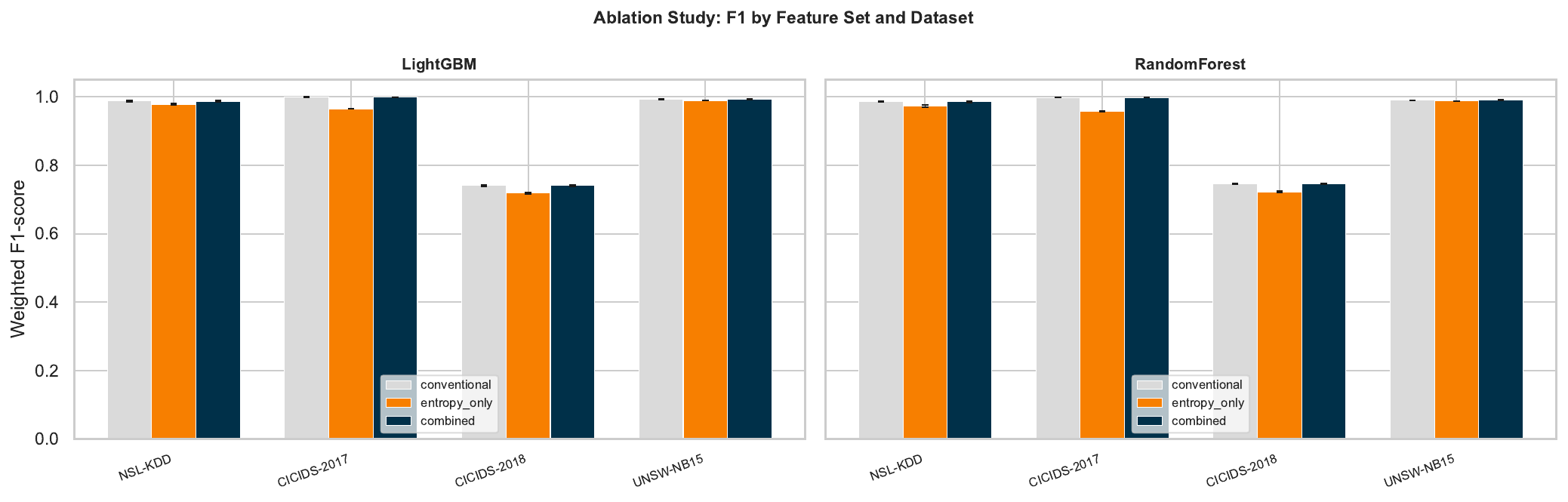}
  \caption{Ablation study: weighted F1-score per feature set and dataset. Gray: conventional features; amber: entropy-only (MDE); navy: combined.}
  \label{fig:ablation}
\end{figure}

\subsection{Temporal Generalization}\label{sec:timesplit}

To evaluate temporal generalization, we train on Monday--Thursday files (150K stratified sample, 12 attack types) and test on the Friday files (100K stratified sample), which contain DDoS and PortScan, attack types that were only partially represented in the training period. Table~\ref{tab:timesplit} reports the results.

\begin{table}[t]
\centering
\caption{CICIDS-2017 temporal evaluation: Mon--Thu training (150K, debiased), Friday test (100K). All metrics weighted except DR and FAR. Bold: best per column.}
\label{tab:timesplit}
\footnotesize
\setlength{\tabcolsep}{2.5pt}
\resizebox{\linewidth}{!}{%
\begin{tabular}{ll rrrrrrrrr}
\toprule
\textbf{Model} & \textbf{Ablation} & F1 & Prec & Rec & DR & FAR & MCC & AUC & PR-AUC \\
\midrule
LightGBM     & conventional  & 0.5082 & 0.7684 & 0.6226 & 0.0816 & \textbf{0.0002} & 0.2220 & 0.8435 & 0.7810 \\
             & entropy\_only & 0.5779 & 0.7232 & 0.6523 & 0.1857 & 0.0223 & 0.2817 & 0.6483 & 0.6202 \\
             & combined      & 0.5080 & 0.7685 & 0.6225 & 0.0814 & \textbf{0.0002} & 0.2219 & \textbf{0.8594} & \textbf{0.7898} \\
\midrule
RandomForest & conventional  & 0.6498 & 0.8017 & 0.7035 & 0.2790 & 0.0005 & 0.4296 & 0.8106 & 0.7665 \\
             & entropy\_only & 0.4305 & 0.3458 & 0.5754 & 0.0002 & 0.0236 & -0.0981 & 0.5740 & 0.5671 \\
             & combined      & \textbf{0.6513} & \textbf{0.8019} & \textbf{0.7044} & \textbf{0.2814} & 0.0006 & \textbf{0.4314} & 0.7997 & 0.7585 \\
\bottomrule
\end{tabular}%
}
\end{table}

Table~\ref{tab:timesplit} shows that all models fail to detect the majority of attacks under distribution shift, and F1 substantially overstates operational utility. LightGBM with conventional features detects fewer than 1 in 12 attacks (DR $=$ 0.082) despite weighted F1 $\approx$ 0.51, because the model defaults to predicting benign. The best result comes from Random Forest (RF) combined (DR $=$ 0.281, MCC $=$ 0.431, FAR $=$ 0.0006), which still misses 72\% of attacks while keeping false alarms near zero. The two entropy-only models behave in opposite ways. LightGBM entropy-only reaches the highest detection rate of any LightGBM condition (DR $=$ 0.186, MCC $=$ 0.282) at FAR $=$ 0.022, whereas RF entropy-only is the clearest failure, with DR $\approx$ 0.000, MCC $= -0.098$, and FAR $=$ 0.024, a model that raises false alarms on benign traffic while detecting almost no attacks, an inversion that F1 $=$ 0.43 alone would not reveal. Ranking quality separates the feature sets. Under LightGBM, AUC is 0.844 for conventional and 0.859 for combined but 0.648 for entropy-only, with PR-AUC 0.781, 0.790, and 0.620, so the entropy features alone rank Friday traffic markedly worse than the conventional features, and adding them to the conventional set changes AUC by $+0.016$. The divergence between threshold metrics (DR, MCC) and ranking metrics (AUC) appears in the conventional and combined sets, where AUC stays above 0.84 while DR stays below 0.09.

\subsection{Pseudo-Live Temporal Replay}\label{sec:replay}

The batch temporal evaluation (Section~\ref{sec:timesplit}) characterizes aggregate performance after training. To assess operational behavior under a realistic inference scenario, we implement a pseudo-live temporal replay. The complete pipeline (imputer, clipper, LightGBM) is trained once on Monday--Thursday traffic (150K flows), all parameters are frozen, and the full Friday traffic (703K flows) is then replayed chronologically without labels, retraining, or threshold adaptation. Friday flows are sorted by their original timestamps and divided into $N=20$ equal windows; metrics are computed retrospectively in each window using ground-truth labels only after replay completes. We compare two threshold strategies: fixed ($\theta=0.5$) and Youden's-J optimal threshold derived from training-set predicted probabilities.

Table~\ref{tab:replay} summarizes overall results across feature conditions. DR collapses to 0.08--0.18 across the conditions, consistent with the batch evaluation, while AUC stays at 0.842 and 0.858 for the conventional and combined models and falls to 0.647 for the entropy-only model. Training-derived threshold recalibration provides no improvement. The Youden threshold (0.705 for combined) is derived from a training distribution where attacks receive high probabilities, but under temporal shift the score distribution drifts downward, placing most attack flows below both training-derived thresholds. DR with $\theta_\text{Youden}$ is within $0.002$ of fixed $\theta=0.5$ in every condition.

\begin{table}[h]
\centering
\caption{Pseudo-live temporal replay summary: LightGBM, full Friday traffic (703K flows, 20 windows). DR and FAR under fixed $\theta=0.5$ and Youden's-J threshold. AUC is threshold-independent.}
\label{tab:replay}
\footnotesize
\setlength{\tabcolsep}{3.5pt}
\resizebox{\linewidth}{!}{%
\begin{tabular}{l r r r r r r}
\toprule
\textbf{Feature set} & $\theta_\text{Youden}$ & DR$_{\theta=0.5}$ & DR$_{\theta_Y}$ & FAR$_{\theta=0.5}$ & MCC & AUC \\
\midrule
Conventional & 0.716 & 0.081 & 0.081 & 0.000 & 0.221 & 0.842 \\
Entropy-only & 0.482 & 0.184 & 0.185 & 0.022 & 0.280 & 0.647 \\
Combined     & 0.705 & 0.081 & 0.081 & 0.000 & 0.221 & \textbf{0.858} \\
\bottomrule
\end{tabular}%
}
\end{table}

The window-by-window breakdown (Figure~\ref{fig:temporal_replay}) reveals a pattern invisible in aggregate results. Windows 3--8 correspond to a DDoS surge in which 46--99\% of flows are attacks. In these windows DR falls to at most 0.005 for every condition while AUC stays at 0.68--0.96 for the conventional and combined models and drops to 0.23--0.50 for entropy-only. The conventional and combined classifiers still rank most attacks above most benign flows, but the predicted scores sit below the decision boundary, so operational detection fails completely, and the entropy-only model loses the ranking as well. Windows 9--14, dominated by PortScan traffic, show partial recovery, with DR 0.10--0.25 for the conventional and combined models and 0.36--0.51 for entropy-only, at AUC 0.80--0.98. The entropy-only model is the one that detects a substantial share of the PortScan phase, consistent with the near-unidirectional signature of that family (Section~\ref{sec:framework}), while it is the weakest during the DDoS surge. Window 15 contains a single class and is omitted. This two-phase structure shows that the AUC-DR divergence of Section~\ref{sec:timesplit} is a flow-period-specific collapse rather than an averaging artifact. During high-volume homogeneous attack bursts, score calibration fails while, for the conventional and combined models, discrimination is partly preserved.

\begin{figure*}[t]
\centering
\includegraphics[width=\textwidth]{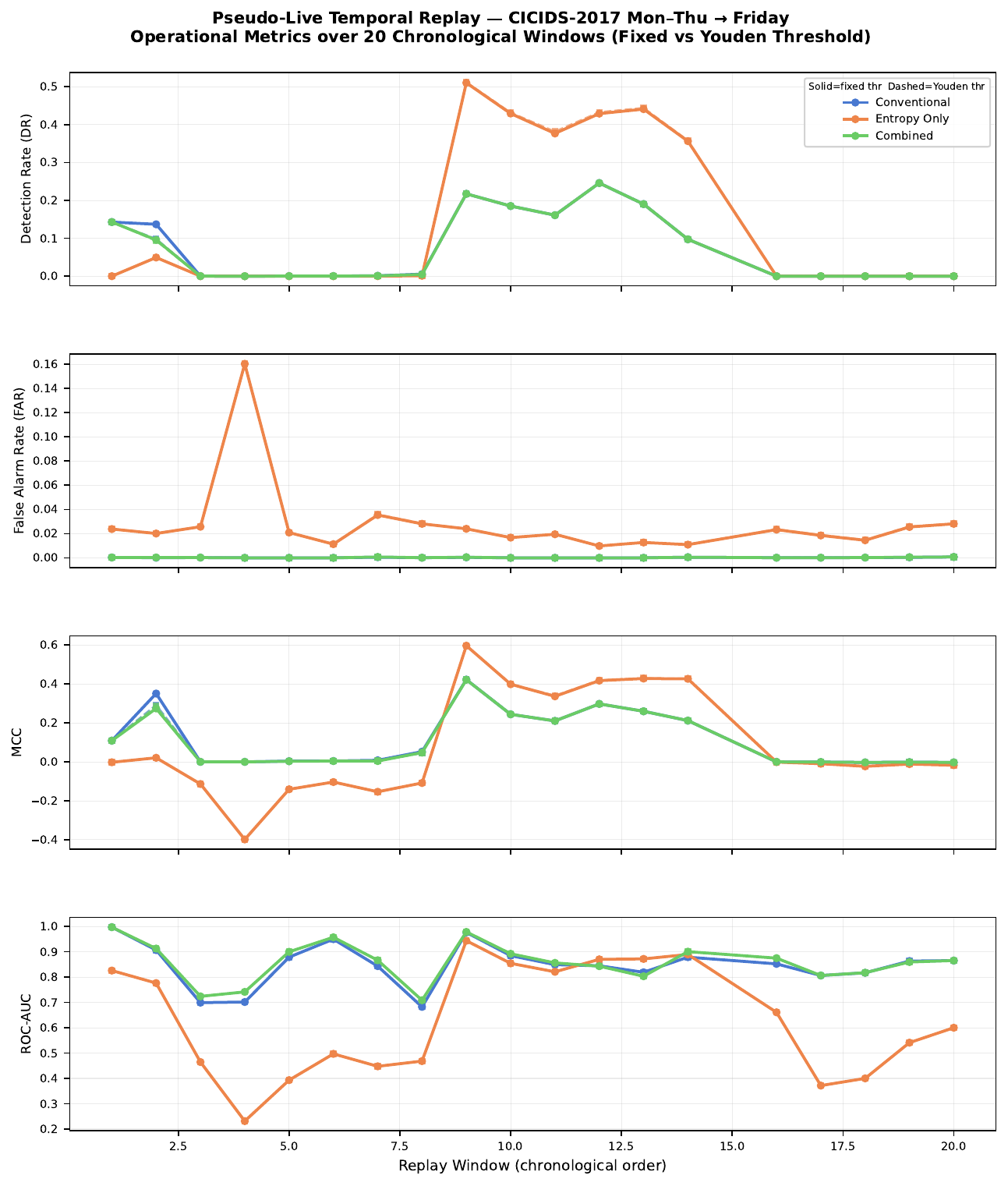}
\caption{Pseudo-live temporal replay: operational metrics over 20 chronological Friday windows (CICIDS-2017). Solid lines: fixed threshold ($\theta=0.5$); dashed lines: Youden's-J threshold. Windows 3--8 correspond to the DDoS surge period; windows 9--14 to PortScan. DR collapses to zero during the surge while AUC reaches 0.95, demonstrating that score calibration fails before discrimination does.}
\label{fig:temporal_replay}
\end{figure*}

\subsection{Baseline Comparison}\label{sec:baselines}

Table~\ref{tab:inpipeline} situates within-distribution performance of MDE against a broad set of classifiers. A fair comparison requires identical training data, identical train/test splits, identical preprocessing, and identical metrics, and every model here meets that condition. Table~\ref{tab:inpipeline} reports 5-fold CV results for seven classifiers on the combined MDE feature set across all four datasets: three gradient-boosted tree methods (LightGBM, XGBoost, CatBoost~\cite{prokhorenkova2018catboost}), Random Forest, and three neural approaches (MLP, TabNet~\cite{arik2021tabnet}, FT-Transformer~\cite{gorishniy2021fttransformer}). All seven train on the full training fold of each split under the identical leakage-free pipeline (fold-local entropy features, imputation, and percentile clipping). LightGBM, Random Forest, CatBoost, and FT-Transformer use balanced class weights, XGBoost uses none, and MLP and TabNet rely on early stopping on a validation split of the training fold. Reporting the full suite of operational metrics, including DR, FAR, MCC, and F1, exposes gaps that aggregate scores alone would conceal. On UNSW-NB15, MLP achieves F1 $=$ 0.942 but DR $=$ 0.613, inflated by the heavily benign majority; TabNet (DR $=$ 0.952) and FT-Transformer, which detects every attack flow at FAR $=$ 0.013, do not show this weakness. Gradient-boosted classifiers (LightGBM, XGBoost, CatBoost) achieve DR $\geq$ 0.96 on NSL-KDD, CICIDS-2017, and UNSW-NB15, with MCC $\geq$ 0.92 on those three datasets. CICIDS-2018 is the stress test, and its intra-class ambiguity limits every approach. No classifier reaches DR $=$ 0.44, MCC stays below 0.41 for all models, and MLP and TabNet detect the fewest attacks (DR $=$ 0.214 and 0.220), with FT-Transformer at DR $=$ 0.382. The models also trade detection for false alarms differently on this dataset. XGBoost, the one model without class weighting, has the highest MCC (0.407) with FAR $=$ 0.009 but DR $=$ 0.250, whereas LightGBM detects the most attacks (DR $=$ 0.438) at FAR $=$ 0.141. XGBoost achieves the highest F1 on NSL-KDD (0.988) and UNSW-NB15 (0.993), LightGBM on CICIDS-2017 (0.999), and Random Forest on CICIDS-2018 (0.735), where the four tree ensembles lie within 0.4 percentage points of each other. FT-Transformer comes within 0.5 to 1.2 percentage points of the best tree ensemble in F1 on every dataset, with DR $=$ 0.997 on CICIDS-2017, at a higher training cost; TabNet trails it on three datasets and matches it on UNSW-NB15.

Tree ensembles lead on every dataset at equal data budget, by 0.5 to 1.2 percentage points of F1 over the best deep model. The lead is consistent with structural properties of pre-aggregated flow records~\cite{grinsztajn2022tree}. Network flow features are heterogeneous in scale and heavily skewed; tree ensembles handle this natively through recursive splitting that is invariant to monotone transformations, requiring no normalization beyond the fold-local percentile clipping applied here. High-order feature interactions (e.g.\ joint packet-size and direction signals) are captured greedily without manual feature construction. A plausible further factor, which these experiments do not test, is that deep architectures gain most from local correlations, spatial regularity, or smooth temporal structure in their inputs, which flow-level statistical summaries largely lack. FT-Transformer, with learned per-feature embeddings and attention over feature tokens, narrows the gap to under one percentage point on NSL-KDD, CICIDS-2017, and UNSW-NB15, consistent with benchmark studies in which deep tabular models have not consistently surpassed tree ensembles on medium-sized tabular data~\cite{grinsztajn2022tree}. The ranking holds for the architectures and settings tested here and is not a general statement about deep models on flow data. The CICIDS-2018 results are illustrative. No deep model reaches the F1 of any tree ensemble or the DR of the best one, and MLP degrades substantially (DR $=$ 0.214 against 0.438 for LightGBM).

\begin{table*}[t]
\centering
\caption{Classifier comparison: seven models on the combined MDE feature set, identical 5-fold stratified CV, identical fold-local pipeline, every model trained on the full training fold. Bold: best F1 per dataset.}
\label{tab:inpipeline}
\footnotesize
\setlength{\tabcolsep}{2.5pt}
\resizebox{\textwidth}{!}{%
\begin{tabular}{ll rrrrrrrr}
\toprule
\textbf{Dataset} & \textbf{Classifier} & F1 & Prec & Rec & DR & FAR & Acc & MCC & PR-AUC \\
\midrule
NSL-KDD     & LightGBM                 & 0.9871 & 0.9871 & 0.9871 & 0.9885 & 0.0147 & 0.9871 & 0.9737 & 0.9997 \\
            & Random Forest            & 0.9855 & 0.9855 & 0.9855 & 0.9839 & 0.0126 & 0.9855 & 0.9704 & 0.9995 \\
            & XGBoost                  & \textbf{0.9879} & 0.9879 & 0.9879 & 0.9894 & 0.0142 & 0.9879 & 0.9752 & 0.9997 \\
            & CatBoost                 & 0.9871 & 0.9871 & 0.9871 & 0.9864 & 0.0120 & 0.9871 & 0.9737 & 0.9996 \\
            & MLP                      & 0.9626 & 0.9629 & 0.9626 & 0.9760 & 0.0550 & 0.9626 & 0.9240 & 0.9655 \\
            & TabNet                   & 0.9719 & 0.9719 & 0.9719 & 0.9761 & 0.0337 & 0.9719 & 0.9427 & 0.9964 \\
            & FT-Transformer           & 0.9807 & 0.9807 & 0.9806 & 0.9804 & 0.0191 & 0.9806 & 0.9607 & 0.9988 \\
\midrule
CICIDS-2017 & LightGBM                 & \textbf{0.9989} & 0.9989 & 0.9989 & 0.9983 & 0.0010 & 0.9989 & 0.9964 & 0.9998 \\
            & Random Forest            & 0.9977 & 0.9977 & 0.9977 & 0.9952 & 0.0017 & 0.9977 & 0.9927 & 0.9993 \\
            & XGBoost                  & 0.9987 & 0.9987 & 0.9987 & 0.9977 & 0.0011 & 0.9987 & 0.9958 & 0.9997 \\
            & CatBoost                 & 0.9982 & 0.9982 & 0.9982 & 0.9990 & 0.0020 & 0.9982 & 0.9944 & 0.9994 \\
            & MLP                      & 0.9449 & 0.9464 & 0.9443 & 0.8870 & 0.0416 & 0.9443 & 0.8291 & 0.8083 \\
            & TabNet                   & 0.9755 & 0.9755 & 0.9756 & 0.9318 & 0.0136 & 0.9756 & 0.9226 & 0.9800 \\
            & FT-Transformer           & 0.9919 & 0.9921 & 0.9918 & 0.9972 & 0.0095 & 0.9918 & 0.9747 & 0.9985 \\
\midrule
CICIDS-2018 & LightGBM                 & 0.7313 & 0.7269 & 0.7409 & 0.4381 & 0.1407 & 0.7409 & 0.3206 & 0.5830 \\
            & Random Forest            & \textbf{0.7348} & 0.7312 & 0.7480 & 0.4175 & 0.1228 & 0.7480 & 0.3278 & 0.5794 \\
            & XGBoost                  & 0.7343 & 0.8118 & 0.7827 & 0.2503 & 0.0092 & 0.7827 & 0.4068 & 0.5807 \\
            & CatBoost                 & 0.7339 & 0.7322 & 0.7511 & 0.3912 & 0.1082 & 0.7511 & 0.3258 & 0.5753 \\
            & MLP                      & 0.7013 & 0.7523 & 0.7525 & 0.2139 & 0.0370 & 0.7525 & 0.2933 & 0.3825 \\
            & TabNet                   & 0.7201 & 0.8048 & 0.7744 & 0.2199 & 0.0088 & 0.7744 & 0.3765 & 0.5299 \\
            & FT-Transformer           & 0.7229 & 0.7215 & 0.7396 & 0.3815 & 0.1204 & 0.7396 & 0.2995 & 0.5483 \\
\midrule
UNSW-NB15   & LightGBM                 & 0.9926 & 0.9929 & 0.9925 & 0.9892 & 0.0073 & 0.9925 & 0.9478 & 0.9950 \\
            & Random Forest            & 0.9900 & 0.9909 & 0.9897 & 0.9973 & 0.0109 & 0.9897 & 0.9316 & 0.9933 \\
            & XGBoost                  & \textbf{0.9933} & 0.9934 & 0.9933 & 0.9608 & 0.0041 & 0.9933 & 0.9517 & 0.9942 \\
            & CatBoost                 & 0.9896 & 0.9906 & 0.9893 & 0.9987 & 0.0115 & 0.9893 & 0.9291 & 0.9932 \\
            & MLP                      & 0.9424 & 0.9434 & 0.9427 & 0.6128 & 0.0310 & 0.9427 & 0.5828 & 0.4309 \\
            & TabNet                   & 0.9882 & 0.9887 & 0.9881 & 0.9521 & 0.0091 & 0.9881 & 0.9164 & 0.9671 \\
            & FT-Transformer           & 0.9882 & 0.9895 & 0.9878 & 1.0000 & 0.0132 & 0.9878 & 0.9203 & 0.9883 \\
\bottomrule
\end{tabular}%
}
\end{table*}

\begin{figure*}[t]
  \centering
  \includegraphics[width=\textwidth]{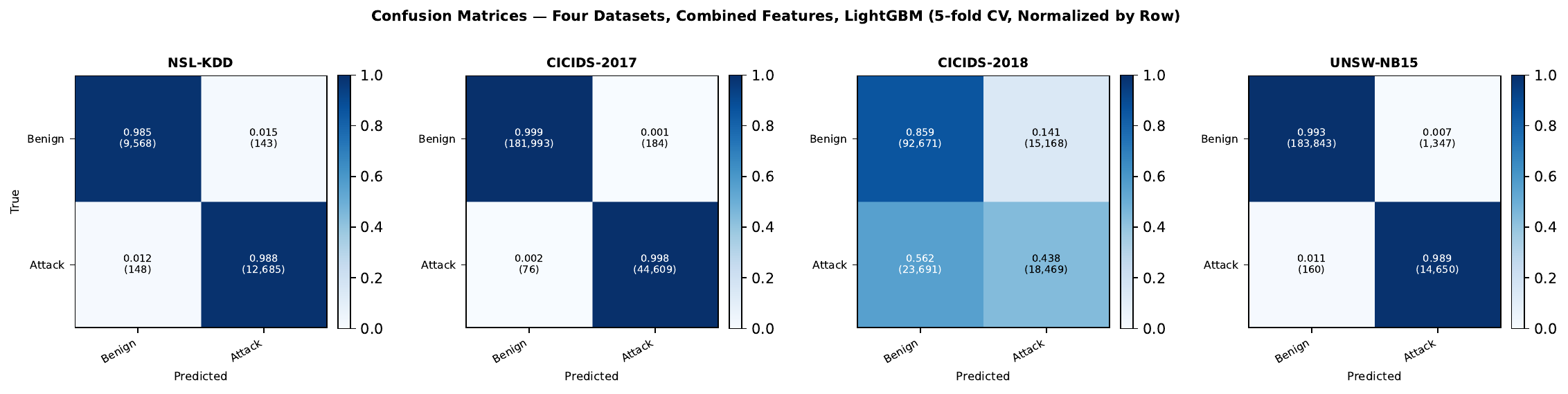}
  \caption{Normalized confusion matrices (rows $=$ true, columns $=$ predicted) for LightGBM with combined features (5-fold CV, pooled). Cells show the row-normalized proportion and raw count. NSL-KDD, CICIDS-2017, and UNSW-NB15 are diagonal-dominant; CICIDS-2018 carries the most off-diagonal mass, with about 56\% of attacks misclassified as benign, in line with its DR $\approx$ 0.44 (Table~\ref{tab:fullmetrics}).}
  \label{fig:cm_all}
\end{figure*}

Table~\ref{tab:baseline} situates MDE-IDS against published methods. Because those studies use different datasets, preprocessing pipelines, and evaluation protocols, the comparison is indicative only; the in-pipeline Table~\ref{tab:inpipeline} is the scientifically fair reference. MDE-IDS provides SHAP-based explainability, which is not reported by any baseline method.

\begin{table}[t]
\centering
\caption{Cross-literature context (indicative only; different protocols). MDE-IDS (LGB) = LightGBM combined; Met.\ = reported metric.}
\label{tab:baseline}
\footnotesize
\setlength{\tabcolsep}{3pt}
\resizebox{\linewidth}{!}{%
\begin{tabular}{l p{3.0cm} l r}
\toprule
\textbf{Dataset} & \textbf{Method} & \textbf{Met.} & \textbf{Score} \\
\midrule
NSL-KDD  & XGBoost~\cite{tama2019indepth}          & F1  & 0.9884 \\
NSL-KDD  & Deep NN~\cite{faker2019intrusion}       & Acc & 0.9910 \\
NSL-KDD  & \textbf{MDE-IDS (LGB)}                 & \textbf{F1}  & \textbf{0.9871} \\
\midrule
CIC-2017 & CNN-LSTM~\cite{liu2022cnn}              & F1  & 0.9780 \\
CIC-2017 & RF~\cite{engelen2021troubleshooting}    & Acc & 0.9997 \\
CIC-2017 & \textbf{MDE-IDS (LGB)}                 & \textbf{F1}  & \textbf{0.9989} \\
\midrule
UNSW-NB15 & Wrap-NN~\cite{kasongo2020deep}         & Acc & 0.9750 \\
UNSW-NB15 & Transformer~\cite{wu2022rtids}         & F1  & 0.9812 \\
UNSW-NB15 & \textbf{MDE-IDS (LGB)}                & \textbf{F1}  & \textbf{0.9926} \\
\bottomrule
\end{tabular}}
\end{table}

\subsection{SHAP Waterfall Analysis}

The attribution structure of MDE features (which entropy features drive predictions and whether attributions are consistent across datasets) is examined below. We first verify leakage-free attribution, then assess cross-dataset rank stability, and finally examine instance-level waterfall plots.

All SHAP figures were generated from the leakage-free pipeline (label column removed from feature matrices before SHAP computation). Feature name audits confirmed that no label-related column appears in any SHAP plot. Each dataset passed a programmatic check verifying that no feature name contains the strings \texttt{label}, \texttt{class}, \texttt{attack}, \texttt{binary}, or \texttt{multi}.

Formal cross-dataset Spearman correlation is not feasible because each dataset's conventional features carry schema-specific column names (e.g., \texttt{Total\_Length\_of\_Fwd\_Packets} in CICIDS vs.\ \texttt{src\_bytes} in NSL-KDD), leaving fewer than three features in common by name across any pair. We therefore report the \emph{MDE top-10 fraction}: the proportion of the ten features with the largest mean absolute SHAP value (LightGBM fitted on the full sample, 600 explained instances) that are MDE-type (L1 ADE, L2 JSD, L3 flag-incidence entropy, or the composite score). Across the four datasets:
\begin{center}
\footnotesize
\setlength{\tabcolsep}{3pt}
\begin{tabular}{lrl}
\toprule
Dataset & Top-10 & Best-ranked MDE features (rank) \\
\midrule
CICIDS-2017 & 2/10 & \texttt{flag\_entropy} (1), \texttt{jsd\_iat} (8) \\
CICIDS-2018 & 1/10 & \texttt{mde\_score} (4) \\
NSL-KDD     & 0/10 & \texttt{conn\_state\_entropy} (14) \\
UNSW-NB15   & 0/10 & \texttt{jsd\_pkt\_sz} (13) \\
\bottomrule
\end{tabular}
\end{center}
\noindent MDE features make up 13\% to 21\% of the feature count and 0\% to 20\% of the top-10 attributions, so they are not disproportionately represented, and on three of the four datasets the strongest evidence of the model comes from conventional features. CICIDS-2017 is the exception, where the flag-incidence entropy has the largest mean absolute SHAP value of all features and the cross-directional JSD of inter-arrival times (\texttt{jsd\_iat}) is eighth. On CICIDS-2018 the only MDE feature in the top ten is the composite score (fourth), with a small absolute attribution, on NSL-KDD and UNSW-NB15 none is, and no ADE feature ranks above seventeenth on any dataset.

Fig.~\ref{fig:waterfall_cicids17} shows the SHAP waterfall plot for the most confidently classified attack instance in CICIDS-2017 (LightGBM, combined features). A waterfall plot decomposes the model output for a single instance. Starting from the expected model output $E[f(X)]$ (base value), each bar shows how a specific feature's value pushes the prediction upward (toward attack) or downward (toward benign), arriving at the final model output $f(x)$. The largest contribution comes from the backward packet-length standard deviation ($+8.1$ in log-odds), followed by the flag-incidence entropy at its minimum value of zero ($+3.2$), which records that every flag incidence in the flow is of a single type, and by the initial window sizes in both directions. No JSD or ADE feature is among the five largest contributions for this instance.

\begin{figure}[t]
  \centering
  \includegraphics[width=\linewidth]{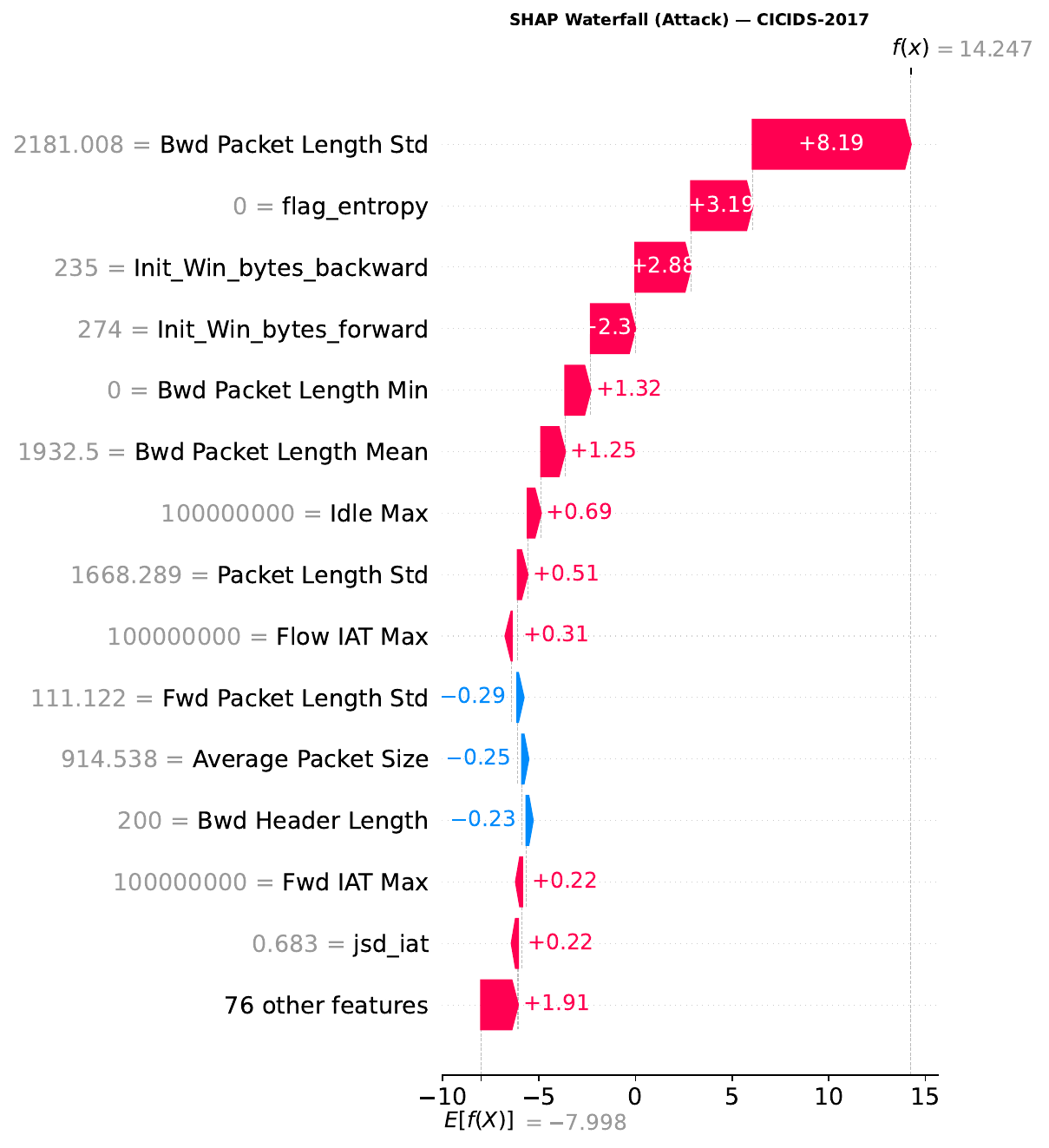}
  \caption{SHAP waterfall plot for a representative attack instance in CICIDS-2017 (LightGBM, combined features). Each bar shows the SHAP contribution of one feature; MDE enters through the flag-incidence entropy (\texttt{flag\_entropy}), while the JSD and ADE features are not among the largest drivers of this prediction.}
  \label{fig:waterfall_cicids17}
\end{figure}

Fig.~\ref{fig:waterfall_unsw} shows the corresponding waterfall for UNSW-NB15. The five largest contributions are all conventional features: \texttt{ct\_state\_ttl} ($+12.9$), \texttt{sttl} ($+6.9$), \texttt{ct\_dst\_sport\_ltm}, \texttt{service}, and \texttt{synack}. No MDE feature appears among them; on this dataset the MDE features rank thirteenth or lower by mean absolute SHAP value, and the model's evidence is the TTL and connection-count structure of the flow.

\begin{figure}[t]
  \centering
  \includegraphics[width=\linewidth]{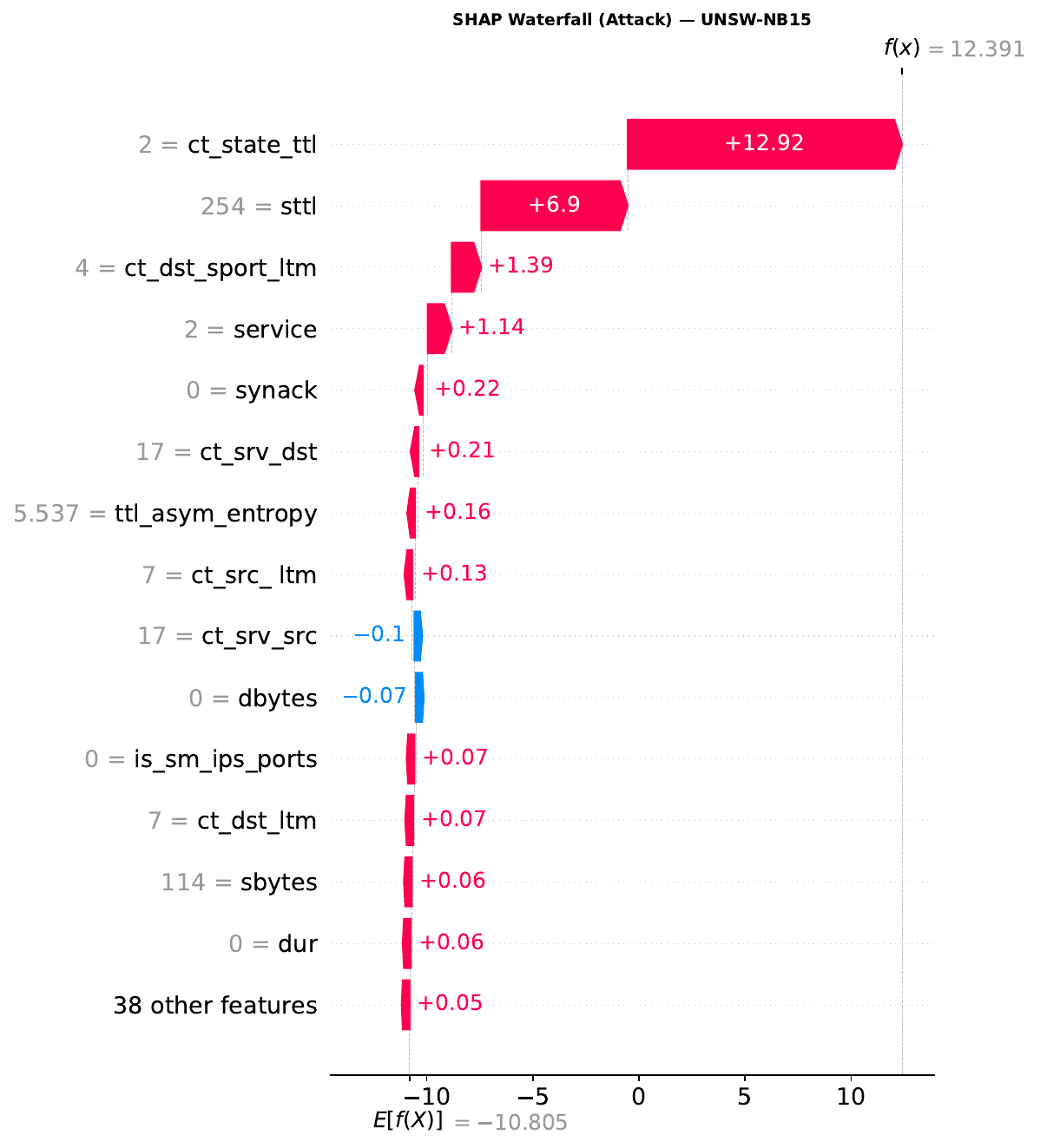}
  \caption{SHAP waterfall plot for a representative attack instance in UNSW-NB15 (LightGBM, combined features). The largest drivers are the connection-state and TTL features; no MDE feature is among them.}
  \label{fig:waterfall_unsw}
\end{figure}

\subsection{Entropy Delta Across Datasets}

The waterfall plots show attribution patterns for individual instances. To characterize the direction and magnitude of the entropy signal at the dataset level, Fig.~\ref{fig:heatmap} presents the entropy delta heatmap. It is the difference in mean MDE score between attack and benign traffic per dataset. Positive delta (attack $>$ benign) indicates that attack flows exhibit \emph{higher} overall distributional entropy; negative delta indicates the reverse. The sign of the delta differs across datasets. Attack flows have a higher composite entropy than benign flows on NSL-KDD ($+0.15$) and CICIDS-2017 ($+0.12$) and a lower one on CICIDS-2018 ($-0.07$) and UNSW-NB15 ($-0.15$). The sign therefore says nothing general about attacks; it reflects which attack families dominate each dataset and how their flows compare with that dataset's own benign traffic, which is why cross-dataset evaluation is necessary.

\begin{figure}[t]
  \centering
  \includegraphics[width=\linewidth]{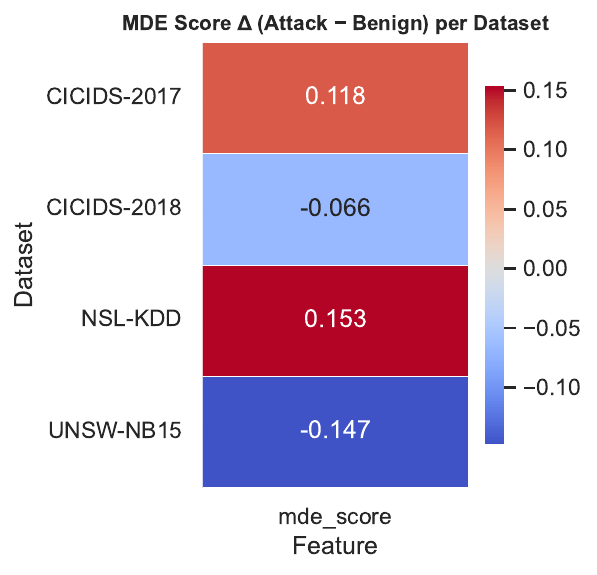}
  \caption{MDE score delta (Attack $-$ Benign mean) per dataset. The sign reversal between environments reveals that attack entropy signatures are environment-dependent, motivating cross-dataset evaluation.}
  \label{fig:heatmap}
\end{figure}

\subsection{Cross-Dataset Transfer}\label{sec:transfer}

To evaluate whether MDE entropy signatures generalize beyond their training distribution, we conduct a zero-shot cross-dataset transfer experiment using the two CIC-FlowMeter-format datasets (CICIDS-2017 and CICIDS-2018), which share an identical entropy feature schema (12 common MDE features). A LightGBM model trained on one dataset's entropy-only features is applied directly to the other dataset, without retraining or fine-tuning. Table~\ref{tab:transfer} presents results alongside in-distribution hold-out baselines.

\begin{table}[h]
\centering
\caption{Zero-shot cross-dataset transfer (entropy-only features, LightGBM). \emph{Hold-out}: 80/20 in-distribution split; \emph{Transfer}: model applied to a different dataset without retraining.}
\label{tab:transfer}
\footnotesize
\setlength{\tabcolsep}{3pt}
\begin{tabular}{l l r r}
\toprule
\textbf{Source} & \textbf{Target} & \textbf{F1} & \textbf{AUC} \\
\midrule
CIC-2017 & CIC-2017 (hold-out) & 0.965 & 0.996 \\
CIC-2017 & CIC-2018 (transfer) & 0.606 & 0.530 \\
\midrule
CIC-2018 & CIC-2018 (hold-out) & 0.720 & 0.670 \\
CIC-2018 & CIC-2017 (transfer) & 0.757 & 0.662 \\
\bottomrule
\end{tabular}
\end{table}

The transfer gap is large in both directions. The CICIDS-2017 model, attuned to the entropy signatures of DDoS, port-scan, and brute-force attacks, drops from a hold-out F1 of 0.965 to 0.606 on CICIDS-2018 with AUC $=$ 0.530 and DR $=$ 0.064, so its scores carry almost no ranking information about CICIDS-2018 attacks. In the reverse direction the CICIDS-2018 model, trained on the harder problem (hold-out F1 $=$ 0.720, AUC $=$ 0.670), reaches F1 $=$ 0.757 and AUC $=$ 0.662 on CICIDS-2017 with DR $=$ 0.490 at FAR $=$ 0.191; it recovers part of the class ordering, but at a false-alarm rate no operator would accept. Entropy signatures are attack-family-specific within the tested environments, and like all discriminative representations they require adaptation when the attack taxonomy changes.

\subsection{Noise Robustness}\label{sec:robustness}

We evaluate the resilience of MDE entropy features to measurement uncertainty by adding zero-mean Gaussian noise $\mathcal{N}(0,\,\varepsilon^2\hat{\sigma}_j^2)$ to each feature $j$ at inference time, where $\hat{\sigma}_j$ is the per-feature standard deviation of the training set and $\varepsilon \in \{0.05, 0.10, 0.25\}$ is the noise fraction. Models are trained on clean data; noise is injected only at test time to simulate sensor degradation or measurement error.

\begin{table}[h]
\centering
\caption{Entropy-only LightGBM under additive Gaussian noise at inference ($\varepsilon$ = fraction of the per-feature standard deviation; $\varepsilon=0$ is the clean baseline). Weighted F1 degrades gradually, whereas DR and AUC show how much detection is lost.}
\label{tab:robustness}
\footnotesize
\setlength{\tabcolsep}{4pt}
\begin{tabular}{l l r r r r}
\toprule
\textbf{Dataset} & & $\varepsilon=0$ & $\varepsilon=0.05$ & $\varepsilon=0.10$ & $\varepsilon=0.25$ \\
\midrule
NSL-KDD   & F1  & 0.978 & 0.891 & 0.858 & 0.813 \\
          & DR  & 0.977 & 0.839 & 0.792 & 0.742 \\
          & AUC & 0.997 & 0.971 & 0.946 & 0.900 \\
\midrule
CIC-2017  & F1  & 0.965 & 0.855 & 0.826 & 0.786 \\
          & DR  & 0.986 & 0.521 & 0.443 & 0.360 \\
          & AUC & 0.996 & 0.811 & 0.743 & 0.663 \\
\midrule
CIC-2018  & F1  & 0.719 & 0.636 & 0.620 & 0.597 \\
          & DR  & 0.382 & 0.234 & 0.262 & 0.312 \\
          & AUC & 0.667 & 0.517 & 0.504 & 0.497 \\
\midrule
UNSW-NB15 & F1  & 0.989 & 0.957 & 0.951 & 0.943 \\
          & DR  & 0.989 & 0.866 & 0.850 & 0.790 \\
          & AUC & 0.999 & 0.984 & 0.978 & 0.964 \\
\bottomrule
\end{tabular}
\end{table}

Weighted F1 degrades gradually, and UNSW-NB15 is the most robust dataset, losing $\Delta=0.046$ at $\varepsilon=0.25$ against $0.166$, $0.122$, and $0.179$ on NSL-KDD, CICIDS-2018, and CICIDS-2017, most of it already at $\varepsilon=0.05$. The detection rate and the AUC show a larger loss (Table~\ref{tab:robustness}). On CICIDS-2017 a perturbation of 5\% of each feature's standard deviation lowers DR from 0.986 to 0.521 and AUC from 0.996 to 0.811, and on CICIDS-2018 AUC falls to 0.497, the level of chance, whereas NSL-KDD and UNSW-NB15 keep DR at 0.742 and 0.790 at $\varepsilon=0.25$. Weighted F1 hides this loss because the benign majority is still classified correctly. The entropy-only models of the two CIC datasets therefore depend on fine differences in feature values that noise of this size erases, and robustness to measurement noise cannot be claimed for these two datasets.

\subsection{Entropy versus Simple Statistics}\label{sec:simplestats}

A necessary test of whether MDE features add genuine information is to compare them against the raw numerical statistics from which they are derived. Table~\ref{tab:simplestats} compares LightGBM with entropy-only MDE features against a \emph{simple-statistics} baseline that uses the corresponding raw columns directly (packet-length means/stds, IAT means/stds, flow rates, byte counts) without any entropy transformation. The number of features in both conditions is comparable ($\leq$12). A third condition removes the log-scaled volume descriptors from the MDE set and keeps the entropy terms, which isolates what the information-theoretic quantities contribute by themselves.

\begin{table}[h]
\centering
\caption{MDE features against the raw statistics they are derived from (LightGBM, 5-fold CV). The second MDE row removes the log-scaled descriptors and keeps the entropy terms. $\Delta$F1 is measured against the simple statistics.}
\label{tab:simplestats}
\scriptsize
\setlength{\tabcolsep}{3pt}
\begin{tabular}{l l r r r r r}
\toprule
\textbf{Dataset} & \textbf{Feature set} & \textbf{$n_f$} & \textbf{F1} & \textbf{DR} & \textbf{FAR} & $\Delta$\textbf{F1} \\
\midrule
NSL-KDD   & Simple statistics  & 11 & 0.983 & 0.986 & 0.021 & \\
          & MDE, all features  &  7 & 0.978 & 0.977 & 0.020 & $-$0.005 \\
          & MDE, entropy terms &  5 & 0.953 & 0.976 & 0.076 & $-$0.029 \\
\midrule
CIC-2017  & Simple statistics  & 12 & 0.977 & 0.981 & 0.025 & \\
          & MDE, all features  & 12 & 0.965 & 0.986 & 0.042 & $-$0.012 \\
          & MDE, entropy terms & 11 & 0.962 & 0.979 & 0.043 & $-$0.015 \\
\midrule
CIC-2018  & Simple statistics  & 12 & 0.689 & 0.488 & 0.238 & \\
          & MDE, all features  & 12 & 0.719 & 0.382 & 0.128 & $+$0.030 \\
          & MDE, entropy terms & 11 & 0.719 & 0.377 & 0.124 & $+$0.031 \\
\midrule
UNSW-NB15 & Simple statistics  & 10 & 0.989 & 0.991 & 0.012 & \\
          & MDE, all features  & 11 & 0.989 & 0.989 & 0.012 & $<$0.001 \\
          & MDE, entropy terms &  9 & 0.988 & 0.988 & 0.012 & $<$0.001 \\
\bottomrule
\end{tabular}
\end{table}

Simple statistics achieve slightly higher standalone F1 on NSL-KDD and CICIDS-2017, where raw packet-count and byte-rate features are direct discriminators for volume-based attacks, and the two feature sets are equivalent on UNSW-NB15. CICIDS-2018 is the exception, where the MDE features outperform the twelve raw statistics they are derived from by 0.030 F1, with a lower false alarm rate (0.128 against 0.238) but also a lower detection rate (0.382 against 0.488). Removing the log-scaled descriptors changes F1 by at most 0.003 on CICIDS-2017, CICIDS-2018, and UNSW-NB15, so the entropy terms carry the result on those datasets. On NSL-KDD it lowers F1 from 0.978 to 0.953 and raises the false alarm rate from 0.020 to 0.076, so part of the NSL-KDD performance of the MDE set belongs to the two log byte counts rather than to the entropy terms. The entropy transformation therefore does not consistently improve standalone F1 over raw input statistics, and where it helps it does so by trading detection for fewer false alarms. The contribution is in representation grounding. Each MDE feature is defined by an information-theoretic formula (Propositions~\ref{prop:ade}--\ref{prop:jsd}), follows one three-level design that is instantiated per schema, with proxies where a schema lacks the inputs, and receives domain-interpretable SHAP attributions with cross-fold rank stability, as shown in Section~\ref{sec:shapstability}.

\subsection{Per-Class Detection}\label{sec:perclass}

Table~\ref{tab:perclass} reports per-category DR and FAR for LightGBM on the combined feature set across three evaluation conditions: (i) random 20\% hold-out on NSL-KDD and CICIDS-2018, and (ii) the CICIDS-2017 temporal split (Mon--Thu train, Friday test). These conditions reveal meaningfully different failure profiles.

\begin{table}[h]
\centering
\caption{Per-category detection rates (LightGBM, combined MDE features). NSL-KDD and CICIDS-2018 use a random stratified 20\% hold-out; CICIDS-2017 uses the temporal split (Mon--Thu train $\to$ Friday test).}
\label{tab:perclass}
\footnotesize
\setlength{\tabcolsep}{4pt}
\begin{tabular}{l l r r r}
\toprule
\textbf{Dataset / Split} & \textbf{Category} & \textbf{$n_{test}$} & \textbf{DR} & \textbf{FAR} \\
\midrule
NSL-KDD (random)     & Attack (binary) & 2,567 & 0.987 & n/a \\
                     & Benign & 1,942 & n/a & 0.012 \\
\midrule
CICIDS-2018 (random) & Attack (binary) & 8,432 & 0.441 & n/a \\
                     & Benign & 21,568 & n/a & 0.140 \\
\midrule
CIC-2017 (temporal)  & DDoS & 18,888 & 0.182 & n/a \\
                     & PortScan & 18,480 & 0.001 & n/a \\
                     & Bot & 329 & 0.000 & n/a \\
                     & Benign & 62,303 & n/a & 0.000 \\
\bottomrule
\end{tabular}
\end{table}

\noindent Under random hold-out, NSL-KDD achieves DR=0.987 and FAR=0.012, consistent with its aggregate F1=0.987. CICIDS-2018 is harder, with DR=0.441 and FAR=0.140 in line with aggregate F1=0.732, so fewer than half of the attacks are detected even under a random split.

The temporal evaluation on CICIDS-2017 reveals the most informative failure pattern. Bot traffic (n=329) is undetected (DR=0.000) because no Bot flows appear in Mon--Thu training; DDoS achieves only DR=0.182 because Friday's DDoS tool profile differs from the Mon--Thu DoS attacks; PortScan is almost entirely missed (DR=0.001) for the same reason. The benign FAR stays at 0.0003, so the model defaults to predicting benign when it finds no matching pattern rather than generating false alarms. This is distribution shift between training and deployment traffic in its plainest form, and it motivates domain-adaptive recalibration as future work.

\subsection{Unseen Attack Families}\label{sec:unseen}

To test MDE beyond random stratified CV, we conduct an unseen-attack-family evaluation on CICIDS-2017. A LightGBM and Random Forest model are trained on all attack families \emph{except} two held-out families (Infiltration and Bot), then evaluated on flows from only those two families plus a matched benign set. This simulates the operationally realistic scenario where a deployed IDS encounters attack types not seen during training.

\begin{table*}[h]
\centering
\caption{Unseen-attack-family evaluation: trained on CICIDS-2017 without Infiltration and Bot, tested on those held-out families plus matched benign ($n_\text{attack}{=}159$, $n_\text{benign}{=}182{,}177$). Prec/Rec are weighted; DR is computed on held-out attacks, FAR on benign. All conventional and combined models achieve DR$=$0.000, confirming genuine distribution shift.}
\label{tab:unseen}
\footnotesize
\setlength{\tabcolsep}{4pt}
\begin{tabular}{l l r r r r r r}
\toprule
\textbf{Model} & \textbf{Ablation} & \textbf{F1} & \textbf{Prec} & \textbf{Rec} & \textbf{AUC} & \textbf{DR} & \textbf{FAR} \\
\midrule
LightGBM     & conventional  & 0.9984 & 0.9983 & 0.9985 & 0.7553 & 0.0000 & 0.0006 \\
             & entropy\_only & 0.9779 & 0.9982 & 0.9584 & 0.6150 & 0.0000 & 0.0408 \\
             & combined      & 0.9984 & 0.9983 & 0.9985 & 0.7726 & 0.0000 & 0.0006 \\
\midrule
RandomForest & conventional  & 0.9981 & 0.9983 & 0.9979 & 0.4307 & 0.0000 & 0.0013 \\
             & entropy\_only & 0.9742 & 0.9982 & 0.9514 & 0.4811 & 0.0000 & 0.0478 \\
             & combined      & 0.9981 & 0.9983 & 0.9979 & 0.6438 & 0.0000 & 0.0012 \\
\bottomrule
\end{tabular}
\end{table*}

DR$=$0.000 for every model and feature condition on held-out Infiltration and Bot traffic, while weighted F1 remains above 0.97, driven by the 99.9\% benign majority. None of the evaluated feature conditions generalized to these held-out families at the deployed decision threshold. Infiltration ($n{=}2$) mimics legitimate HTTP sessions; Bot ($n{=}157$) uses low-frequency C\&C channels; both are structurally unlike the DoS and scan traffic in training. The entropy-only condition detects none of the 159 held-out flows either and adds false alarms (FAR$=$0.041--0.048), so it offers nothing operationally here. LightGBM conventional achieves AUC$=$0.755 versus RF's 0.431, showing that the ranking function retains partial ordering even when threshold-level decisions fail completely. This gap between DR and AUC is a methodological caution. High AUC does not imply operational detection capability when new attack families are present.

\subsection{SHAP Fold Stability}\label{sec:shapstability}

To quantify the consistency of SHAP attributions across cross-validation folds, we compute, for each of the 5 CV folds: (i) the mean absolute SHAP value per feature on the held-out test fold, (ii) the Spearman rank correlation between each pair of folds, and (iii) the Kendall $\tau$ for robustness. Table~\ref{tab:shapstability} reports results for three representative datasets (NSL-KDD, CICIDS-2017, UNSW-NB15).

\begin{table}[h]
\centering
\caption{SHAP rank stability across 5 CV folds (LightGBM, combined features). $\rho_\text{all}$ / $\rho_\text{MDE}$ = Spearman rank correlation over all features / MDE-only features, mean $\pm$ std across $\binom{5}{2}{=}10$ fold pairs. $\tau_K$ = Kendall's $\tau$ (all features). MDE $\sigma_r$ = mean rank-position std of MDE features across folds.}
\label{tab:shapstability}
\footnotesize
\setlength{\tabcolsep}{4pt}
\begin{tabular}{l r r r r}
\toprule
\textbf{Dataset} & $\rho_\text{all}$ & $\rho_\text{MDE}$ & $\tau_K$ & MDE $\sigma_r$ \\
\midrule
NSL-KDD     & $0.940 \pm 0.014$ & $0.571 \pm 0.226$ & $0.830 \pm 0.022$ & 3.87 \\
CICIDS-2017 & $0.914 \pm 0.018$ & $0.656 \pm 0.154$ & $0.797 \pm 0.025$ & 8.95 \\
UNSW-NB15   & $0.845 \pm 0.046$ & $0.649 \pm 0.191$ & $0.705 \pm 0.049$ & 5.50 \\
\bottomrule
\end{tabular}
\end{table}

All-feature SHAP rankings are consistent across folds, with Spearman $\rho$ of 0.94 (NSL-KDD), 0.91 (CICIDS-2017), and 0.84 (UNSW-NB15) and Kendall $\tau$ of 0.70 to 0.83. MDE-only rankings are more variable ($\rho=0.57$--$0.66$), reflecting the small feature count ($n=7$--$12$) and the sensitivity of pairwise Spearman correlation to rank swaps among similarly valued features. The MDE rank std (3.9--8.9 positions) quantifies this within-dataset variability. The aggregate placement of the MDE features is stable across folds, but the ordering among individual MDE features varies, as expected when correlated entropy features (ADE and JSD share packet-size inputs) exchange positions without changing the overall profile.

\subsection{Computational Profiling}\label{sec:profiling}

Table~\ref{tab:profiling} reports training time and per-flow inference latency for LightGBM and Random Forest on the combined feature set, measured on a stratified 80/20 split with the tree ensembles restricted to one CPU core. Both timings cover the whole fitted pipeline rather than the classifier alone, and the inference latency is further split into the feature stage, which constructs the entropy features from the flow statistics and applies the training-fitted median imputation and percentile clipping, and the classifier prediction. A budget of 100~$\mu$s per flow, equivalent to 10K flows per second on a single core, is taken as the viability criterion for a software monitoring appliance.

\begin{table}[!ht]
\centering
\caption{Training time and per-flow inference latency of the complete pipeline on the combined feature set, single CPU core, stratified 80/20 split. Inference is split into the feature stage (entropy features, imputation, clipping) and the classifier; the two stages are timed separately and do not sum exactly to the total. The 100~$\mu$s per flow budget (10K flows/s on one core) is the viability criterion, and every configuration stays below it.}
\label{tab:profiling}
\scriptsize
\setlength{\tabcolsep}{2.5pt}
\begin{tabular}{l l r r r r r}
\toprule
\textbf{Dataset} & \textbf{Model} & $n_{train}$ & \textbf{Train (s)} & \multicolumn{3}{c}{\textbf{Inference ($\mu$s/flow)}} \\
\cmidrule(lr){5-7}
 & & & & \textbf{Total} & \textbf{Feat.} & \textbf{Clf.} \\
\midrule
NSL-KDD   & LightGBM     & 18,035 &   1.48 & 28.46 &  6.33 & 19.12 \\
          & RandomForest & 18,035 &   2.33 & 22.09 &  5.36 & 16.31 \\
\midrule
CIC-2017  & LightGBM     & 181,489 &  18.14 & 18.77 &  4.22 & 14.03 \\
          & RandomForest & 181,489 &  62.34 & 22.22 &  4.11 & 17.50 \\
\midrule
CIC-2018  & LightGBM     & 119,999 &   8.81 & 20.25 &  3.84 & 14.89 \\
          & RandomForest & 119,999 &  64.47 & 37.26 &  3.87 & 35.88 \\
\midrule
UNSW-NB15 & LightGBM     & 160,000 &   7.42 &  8.58 &  2.26 &  7.05 \\
          & RandomForest & 160,000 &  30.93 &  9.48 &  2.94 & 10.37 \\
\bottomrule
\end{tabular}
\end{table}

\noindent Single-core inference takes between 8.6 and 37.3~$\mu$s per flow across the eight configurations, which corresponds to roughly 27K to 117K flows per second on one core and leaves a margin of between 2.7 and 11.7 times under the 100~$\mu$s budget. The feature stage costs 2.3 to 6.3~$\mu$s per flow and accounts for 10\% to 31\% of the total; the largest classifier cost is that of Random Forest on CICIDS-2018 (35.9~$\mu$s). Each MDE feature is an explicit function of the flow statistics and the JSD features are read from a precomputed table, so the feature stage scales linearly with the number of flows and can be parallelized across cores. Training time grows with the size of the training set, from 1.5 to 18.1 seconds for LightGBM and from 2.3 to 64.5 seconds for Random Forest on 18K to 181K flows. These figures are upper bounds for a deployment, since a production build would use all available cores.

\section{Discussion}\label{sec:discussion}
\vspace{0.5em}

Two claims run through the results and are worth separating. The entropy transformation contributes representation grounding rather than predictive uplift, and the evaluation protocol exposes operational failure modes that aggregate F1 cannot. The following subsections take each in turn and then bound the claims.

\subsection{Contribution and Interpretability}

Combined and conventional conditions achieve comparable within-distribution F1 across all datasets (Table~\ref{tab:ablation}); the entropy transformation does not consistently improve discriminative power over conventional features. The contribution is therefore in representation quality and evaluation methodology. Entropy features are grounded in specific behavioral claims (Propositions~\ref{prop:ade}--\ref{prop:jsd}), follow one three-level design that is instantiated per schema without retraining, and receive reproducible SHAP attributions, although sizeable attributions are limited to the flag-incidence entropy on CICIDS-2017, and the JSD and ADE features receive small ones. The leakage-free fold-local protocol with full operational metrics (DR, FAR, MCC, PR-AUC) surfaces failure modes that aggregate F1 alone would conceal.

The fold-stability analysis (Section~\ref{sec:shapstability}) shows that SHAP rankings are reproducible across data splits, with Spearman $\rho$ of 0.84 to 0.94 and Kendall $\tau$ of 0.70 to 0.83 for the all-feature models. The entropy features are used but not dominant. The flag-incidence entropy has the largest mean absolute SHAP value of all features on CICIDS-2017, where the cross-directional JSD also enters the top ten (eighth); on CICIDS-2018 the composite score (fourth) is the only MDE feature among the ten strongest contributors, on NSL-KDD and UNSW-NB15 none is, and no ADE feature ranks above seventeenth on any dataset. Stable attribution confirms that the models rely on the same features fold after fold, though SHAP does not independently verify the Gaussian distributional assumptions. The MDE-only rankings show higher variance ($\rho_\text{MDE}=0.57$--$0.66$), expected given the small feature count and the correlation structure (ADE and JSD share packet-size inputs).

\subsection{Operational Failure Modes}

The temporal and unseen-family evaluations reveal two failure patterns that aggregate F1 conceals. Under temporal distribution shift, the conventional and combined models keep part of their score ordering (AUC$=$0.84--0.86) while fixed thresholds collapse (DR$=$0.08), a pattern termed \emph{threshold-ranking divergence}, which points to adaptive recalibration as the practical path forward, whereas the entropy-only model loses the ordering as well (AUC$=$0.65). Under unseen attack families, DR$=$0 for all feature conditions regardless of model; F1$>$0.998 is driven entirely by the majority-class prior. Neither failure is specific to MDE, since the conventional and combined conditions fail in the same way in both experiments, in line with the transfer failures reported for conventional flow features~\cite{sarhan2021netflow}; for the evaluated models, operation under these shifts would require training data from the target period and threshold recalibration.

The cross-dataset transfer gap points in the same direction. The CICIDS-2017 model drops from a hold-out F1 of 0.965 to 0.606 on CICIDS-2018 with AUC of 0.530, and the reverse transfer reaches F1 of 0.757 only at a false-alarm rate of 0.191. Entropy signatures are attack-family-specific within the tested environments. Transfer of the conventional features was not measured here, and prior work reports large cross-dataset gaps for them as well~\cite{sarhan2021netflow}, so the gap is not shown to be unique to entropy.

\subsection{Scope and Limitations}\phantomsection\label{sec:limitations}

The attributions of Section~\ref{sec:shapstability} are predicated on the Gaussian approximation underlying ADE and JSD. Flow statistics exhibit substantial non-Gaussianity (skewness 1.5--9.7, excess kurtosis 0.9--110.2), so ADE is the maximum-entropy value for the observed variance, an upper bound on the true entropy (Remark~\ref{rem:upper}) rather than an estimate of it, and it carries no information about distribution shape for multimodal and heavy-tailed flows. Its discriminative value does not follow from the bound. The negentropy that separates the Gaussian value from the true entropy can differ between attack and benign flows, so the relative ordering is preserved only where those gaps are comparable, and the evidence that the ordering is useful is empirical: the ablation results, the noise-robustness experiments, and the comparison with the raw statistics from which ADE is derived. Under Gaussian noise of 5\% of each feature's standard deviation the detection rate of the entropy-only model falls from 0.99 to 0.52 on CICIDS-2017 while it stays at 0.87 on UNSW-NB15, so the margin between inter-class entropy differences and perturbations of the features is dataset-dependent; non-parametric estimators (kernel density, R\'enyi) are natural successors where that margin is small.

Four further limits bound the claims. On UNSW-NB15 and NSL-KDD several MDE terms are proxies rather than the quantities the levels are defined on (Table~\ref{tab:inventory}), and on NSL-KDD the log-scaled descriptors account for part of the entropy-only performance (Section~\ref{sec:simplestats}). Entropy signatures require recalibration when attack taxonomies change, as the zero-shot transfer F1 of 0.61--0.76 and the unseen-family DR$=$0 show. Standard stratified CV allows campaign-correlated flows across train and test folds; host-based or session-level isolation would offer stricter distribution guarantees than temporal splits alone. All results are specific to four controlled benchmark datasets, and validation on live organizational traffic is required before any operational deployment claim can be made.

\section{Conclusion}\label{sec:conclusion}
\vspace{0.5em}

MDE derives 7--12 interpretable features, most of them entropy terms, from pre-aggregated flow statistics, requiring no raw packet access, learned parameters, or architecture choices; only the imputation medians and the composite-score bounds are fitted on the training split. Within-distribution cross-validation performance of the combined and conventional features differs by at most 0.1 percentage points of F1, so the entropy transformation does not degrade aggregate discriminative power; its contribution is in representation grounding and evaluation methodology.

Three evaluation protocols expose failure modes that aggregate F1 structurally cannot reveal. First, full operational metric reporting shows that CICIDS-2018 F1$=$0.73 masks a DR of only 0.44, so fewer than half of attacks are detected, a gap that F1 alone hides under class imbalance. Second, an unseen-attack-family experiment on CICIDS-2017 demonstrates that F1 exceeds 0.998 while DR$=$0 for held-out Infiltration and Bot families; the aggregate score is driven entirely by the 99.9\% benign majority and provides no signal about detection capability on new threats. Third, a pseudo-live temporal replay evaluation, in which the trained pipeline is frozen and 703K Friday flows are replayed chronologically without retraining, identifies \emph{threshold-ranking divergence}, in which AUC remains 0.84--0.86 for the conventional and combined models while DR collapses to 0.08, and the entropy-only model loses ranking as well (AUC 0.65); Youden's-J threshold recalibration provides no recovery because the attack score distribution shifts below any threshold derived from training data. Window-level analysis shows that during DDoS surge periods DR falls to at most 0.005 while the AUC of the conventional and combined models stays as high as 0.96, so score calibration fails independently of score discrimination.

SHAP fold-stability analysis (Spearman $\rho=0.84$--$0.94$ over all features) confirms that attributions are reproducible across folds; within MDE, the flag-incidence entropy carries the largest attributions, on CICIDS-2017, while the JSD and ADE features receive small attributions.

Future directions include non-parametric entropy estimation for heavy-tailed and encrypted flows, learned threshold recalibration for taxonomy-adaptive deployment, campaign-aware cross-validation, and validation on live organizational traffic.

\vspace{0.8em}

\section*{Declarations}
\vspace{0.5em}
\begin{itemize}
  \item \textbf{Funding:} Not applicable.
  \item \textbf{Conflict of Interest:} The authors declare no competing interests.
  \item \textbf{Availability of Data and Materials:} All four datasets are publicly available (NSL-KDD, CICIDS-2017, CICIDS-2018, UNSW-NB15). The full experimental pipeline, including preprocessing, MDE feature extraction, model training, and SHAP analysis, is available as open-source code at \url{https://github.com/drbouke/mde}.
  \item \textbf{Ethics Approval:} Not applicable.
\end{itemize}


\end{document}